\documentclass[lettersize,journal]{IEEEtran}

\usepackage[style=ieee, sorting=none, backend=biber, doi=false]{biblatex} %
\usepackage{amsmath,amssymb,amsfonts,amsthm}

\usepackage{graphicx}
\usepackage{textcomp}
\usepackage{xcolor}
\def\BibTeX{{\rm B\kern-.05em{\sc i\kern-.025em b}\kern-.08em
    T\kern-.1667em\lower.7ex\hbox{E}\kern-.125emX}}

\usepackage{url}

\usepackage{booktabs}       %
\usepackage[flushleft]{threeparttable}
\usepackage{amsfonts}       %
\usepackage{nicefrac}       %
\usepackage{microtype}      %
\usepackage{lipsum}
\usepackage{fancyhdr}       %
\usepackage{physics}
\usepackage{mathtools}
\usepackage{xcolor}
\usepackage{soul}
\usepackage{multirow}
\usepackage{textcomp}

\usepackage{enumitem}

\usepackage{algorithm}
\usepackage[noend]{algpseudocode}

\usepackage{mathabx}
\allowdisplaybreaks

\usepackage{hyperref}       %
\hypersetup{
    colorlinks=true,
    linkcolor=blue,
    filecolor=magenta,
    urlcolor=cyan,
}

\usepackage{soul}

\usepackage{lipsum}

\usepackage{lineno}

\usepackage{setspace}

\newtheorem{remark}{Remark}
\newtheorem{theorem}{Theorem}
\newtheorem{lemma}{Lemma}
\newtheorem{definition}{Definition}

\newcommand{\Iset}{[\![I]\!]}
\newcommand{\Tset}{[\![T]\!]}
\newcommand{\Sset}{[\![S]\!]}
\def\R{\mathbb{R}}

\newcommand{\X}{\mathcal{X}}
\newcommand{\I}{\mathcal{I}}

\newcommand{\Ind}{\mathcal{I}}

\newcommand{\Jgrid}{J^{\text{grid}}}
\newcommand{\Jev}{J^{\text{EV}}}

\title{\LARGE
Scalable Optimization for \textbf{\textsf{M}}obility-\textbf{\textsf{A}}ware \textbf{\textsf{C}}oordinated\\Electric Vehicle Charging in Power Distribution Networks
}

\author{Yi Ju$^{1}$, Lunlong Li$^{2}$, Jingchun Wang$^{1}$, and Scott Moura$^{1}$%
\thanks{$^{1}$Yi Ju, Jingchun Wang, and Scott Moura are with the Department of Civil and Environmental Engineering, University of California, Berkeley.
{\tt\footnotesize \{juy16thu, jingchun\_wang, smoura\}@berkeley.edu}.}%
\thanks{$^{2}$Lunlong Li is with the Department of Civil and Environmental Engineering, the Hong Kong University of Science and Technology.
{\tt\footnotesize lunlong.li@connect.ust.hk}.}%
\thanks{\emph{Manuscript under review}. Last updated: \today. A preliminary version \cite{ju2026admm} was presented at 2026 American Control Conference (ACC).}
}

\begin{document}

\maketitle
\thispagestyle{empty}
\pagestyle{empty}

\begin{abstract}

Rapid growth in electric-vehicle (EV) charging demand is placing increasing stress on power distribution networks (PDNs), whose hosting capacity is often limited and spatially uneven.
Beyond demonstrating that coordination can help, this paper answers an open question that is central for planners: \emph{what is the maximal achievable benefit of EV charging demand flexibility from spatial and temporal shifting in reducing overload-driven distribution upgrades at a regional scale}?

We introduce \textsf{MAC} (\underline{\textsf{M}}obility-\underline{\textsf{A}}ware \underline{\textsf{C}}oordinated EV charging) to establish a credible upper bound, which entails rethinking charging flexibility around individual mobility, fusing travel itineraries with feeder-level hosting-capacity data, and solving population-scale optimization with spatio-temporal coupling to certified near-optimality.
\textit{(i)} MAC expands feasible scheduling by coupling charging decisions over the full mobility horizon. Instead of enforcing per-session energy recovery, it only requires the EV state-of-charge (SOC) to remain sufficient for upcoming trips.
\textit{(ii)} MAC is computationally scalable via an iterated price response (IPR) scheme. Each iteration posts a locational-temporal price, collects the fleet's best responses in parallel, and updates the price from the observed capacity shortage. The solution is therefore deployable by design, and every iteration carries a certified anytime bound on global optimality. Custom batched subproblem solvers remove the per-iteration bottleneck of solving millions of best responses.

In a future-oriented $30\%$ EV adoption scenario for the San Francisco Bay Area, MAC almost eliminates overload-driven upgrade needs relative to unmanaged charging. Comparing across the baseline spectrum, mobility-aware flexibility alone removes most of the overload, outperforming even fully coordinated session-based charging, and coordination on top suppresses most of the remainder. Both levers are thus essential, and the resulting best-case benchmarks provide references for PDN planning and operations.

\begin{figure*}[h!]
    \centering
    \includegraphics[width=0.94\textwidth]{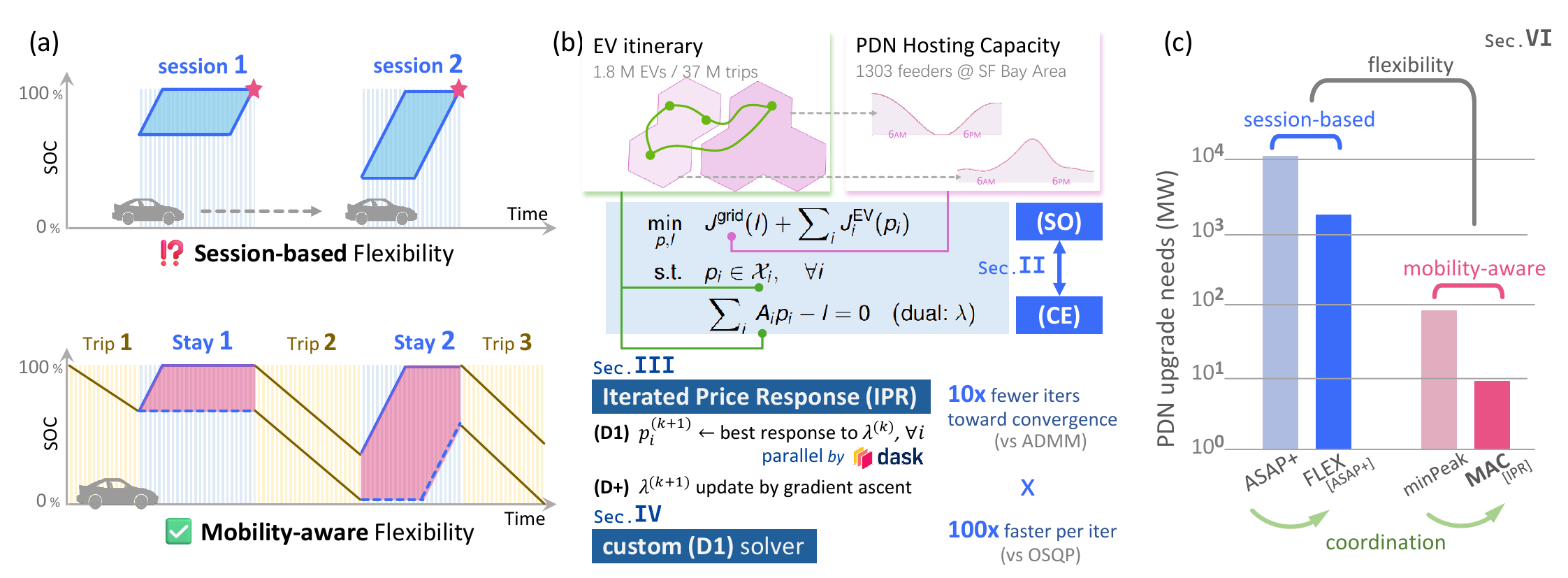}
    \caption{\textbf{Graphic highlight}
    \scriptsize
    \quad
    \textbf{a.} \emph{Conceptualization}: Blue / pink shaded areas are the envelope of feasible charging trajectories under session-based / mobility-aware flexibility.
    \quad
    \textbf{b.} \emph{Methods}: Million-scale EV itineraries and feeder hosting capacities define the coordination problem (SO), implementable as a competitive equilibrium (CE). The iterated price response (IPR) method solves it by alternating fleet-wide best responses (D1) with price updates (D+).
    \quad
    \textbf{c.} \emph{Results}: PDN upgrade needs (total max-violation) across the flexibility and coordination spectrum.
    }
    \label{fig:G-H}
\end{figure*}

\end{abstract}

\section*{Nomenclature}

\noindent\textit{Acronyms}\\[2pt]
{\footnotesize\noindent
\textbf{ADMM}: alternating direction method of multipliers \quad
\textbf{CE}: competitive equilibrium \quad
\textbf{DER}: distributed energy resource \quad
\textbf{DR}: demand response \quad
\textbf{EV}: electric vehicle \quad
\textbf{HC}: hosting capacity \quad
\textbf{IPR}: iterated price response \quad
\textbf{KKT}: Karush--Kuhn--Tucker (optimality conditions) \quad
\textbf{MAC}: mobility-aware coordinated (EV charging) \quad
\textbf{PDN}: power distribution network \quad
\textbf{SO}: social optimum \quad
\textbf{SOC}: state of charge\par}

\smallskip
\noindent\textit{Sets \& Indices}\\[2pt]
{\footnotesize
\begin{tabular}{@{}p{0.18\columnwidth}p{0.75\columnwidth}@{}}
\toprule
Symbol & Description \\
\midrule
$i \in \Iset$ & Index of EVs; $I$ = total number of EVs \\
$t \in \Tset$ & Index of time slots; $T$ = horizon length \\
$s \in \Sset$ & Index of distribution feeders; $S$ = total number of feeders \\
$\I_{s,t}$ & Set of EVs within feeder $s$ service region at time $t$ \\
\bottomrule
\end{tabular}
}

\smallskip
\noindent\textit{Decision Variables}\\[2pt]
{\footnotesize
\begin{tabular}{@{}p{0.18\columnwidth}p{0.75\columnwidth}@{}}
\toprule
Symbol & Description \\
\midrule
$p_i \in \R^T$ & Charging profile (charging power) of EV $i$ \\
$e_{i,t}$ & Battery energy level of EV $i$ at end of slot $t$ \\
$l_s \in \R^T$ & Aggregate EV charging load on feeder $s$ \\
$z_i \in \R^{S \times T}$ & Auxiliary (ADMM splitting) variable; $z_i = A_i p_i$ \\
\bottomrule
\end{tabular}
}

\smallskip
\noindent\textit{Parameters \& Constraints}\\[2pt]
{\footnotesize
\begin{tabular}{@{}p{0.18\columnwidth}p{0.78\columnwidth}@{}}
\toprule
Symbol & Description \\
\midrule
$\Delta t$ & Length of a time slot \\
$E_{i,0}$ & Initial battery energy of EV $i$ \\
$\overline{P}_{i,t},\; \underline{P}_{i,t}$ & Upper/lower charging power bounds for EV $i$ at time $t$ \\
$\overline{E}_{i,t},\; \underline{E}_{i,t}$ & Upper/lower battery energy bounds for EV $i$ at time $t$ \\
$D_{i,t}$ & Driving energy consumption of EV $i$ in slot $t$ \\
$\X_i \subset \R^T$ & Mobility-aware feasible charging set for EV $i$ \\
$C_{s,t}$ & Hosting capacity of feeder $s$ at time $t$ \\
$v_s(l_s)$ & Max load violation on feeder $s$\\
$\Gamma \in \Sset^{I \times T}$ & $\Gamma_{i,t} = s$ means EV $i$ is at feeder $s$ in slot $t$ \\
$A_i \in \quad \{0, 1\}^{\scriptscriptstyle S \times T \times T}$ & Binary routing tensor for EV $i$; $A_i[s, t, t^\prime] =1$ iff $\Gamma_{i,t} = s$ and $t=t^\prime$ \\
\bottomrule
\end{tabular}
}

\smallskip
\noindent\textit{Objective Functions}\\[2pt]
{\footnotesize
\begin{tabular}{@{}p{0.18\columnwidth}p{0.75\columnwidth}@{}}
\toprule
Symbol & Description \\
\midrule
$\Jgrid(l)$ & Grid-side cost (total feeder violations) \\
$\Jev_i(p_i)$ & EV-side disutility; $\Jev_i(p_i) = \frac{\kappa}{2}\norm{p_i}_2^2$ \\
$\kappa$ & Regularization parameter for strong convexity \\
\bottomrule
\end{tabular}
}

\smallskip
\noindent\textit{Dual Variables \& Algorithm Parameters}\\[2pt]
{\footnotesize
\begin{tabular}{@{}p{0.18\columnwidth}p{0.75\columnwidth}@{}}
\toprule
Symbol & Description \\
\midrule
$\lambda \in \R^{S \times T}$ & Dual variable (locational price signal) \\
$Q(\lambda)$ & Dual function of (SO); lower bound on the optimum \\
$\Lambda$ & Dual-feasible price set (product of capped simplices) \\
$N_{s,t}$ & Number of EVs within feeder $s$ at time $t$; $N_{s,t} = |\I_{s,t}|$ \\
$\rho$ & ADMM penalty parameter \\
$\mu_i \in \R^{S \times T}$ & Scaled dual variable (ADMM); $\mu_i = \lambda_i / \rho$ \\
$r_i \in \R^{S \times T}$ & Primal residual (ADMM); $r_i = A_i p_i - z_i$ \\
\bottomrule
\end{tabular}
}

\section{Introduction}

\subsection{Background}
\emph{Hosting capacity} (HC) is a key concept used for evaluating the capability of EV (and other new load or generation) integration into distribution networks \cite{fatima2020review}. Formally, HC is defined as the maximum amount of additional load (or generation) that can be installed at a location without violating operational constraints, deteriorating power quality, or requiring significant infrastructure reinforcements. HC is time-varying depending on base load scenarios and network typology. Utilities such as Pacific Gas and Electric Company (PG\&E) assess HC by running power flow simulations based on their best available data at the time, and publish the analyses regularly \cite{pge_ica_map}.

Studies consistently indicate that substantial distribution-network upgrades will be required to avoid overloading risks as EV adoption grows.
Brockway \textit{et al.} \cite{brockway2021inequitable} find that 39\% of households in California PG\&E territory today lack adequate hosting capacity for even the least power-intensive electrified end uses, such as L1 EV chargers (adding 3 - 5 miles per hour for EVs). Moreover, hosting capacity is unevenly distributed and declines in increasingly Black-identifying and disadvantaged communities, suggesting that grid constraints may exacerbate inequitable access to clean energy options.
Looking forward, Priyadarshan \textit{et al.} \cite{crozier2025distribution} estimate a nationwide 100\% home and vehicle electrification could cost $350 - $790 billion US dollars for reinforcing distribution grids across the US. EV charging is likely to impose significantly larger feeder-capacity impacts than residential space and water-heating electrification \cite{elmallah2022can, mai2018electrification}.
Under the projected EV growth trajectory and unmanaged EV charging behaviors, Li \textit{et al.} \cite{li2024impact} estimate that approximately 50\% of feeders in California by 2035 and 67\% by 2045 would need upgrades, totaling about 25 GW in capacity expansion at an estimated cost of 6-20 billion US dollars.

\subsection{Related Works}
Although demand-side flexibility has been widely discussed in the electrification literature \cite{hardman2025demand}, findings diverge on how much it can defer or reduce distribution infrastructure upgrades, which stems from differences in system objectives, control architectures, and behavioral and operational assumptions.
For example, Elmallah \textit{et al.} \cite{elmallah2022can} evaluate a stylized demand-response (DR) strategy that smooths evening demand peaks by evenly distributing them overnight. They find that the resulting reduction in required distribution upgrades is only marginal in a deep electrified scenario.
In contrast, Navidi \textit{et al.} \cite{navidi2023coordinating} show that coordination of distributed energy resources (DERs) can substantially alleviate distribution constraint violations and reduce peak load (e.g., by 17\% by 2050 in their California case studies).
Similar flexibility opportunities have also been reported in rapidly electrifying contexts such as Chinese metropolises \cite{li2025unlocking}.

Methodologically, a wide range of simulation- and optimization-based approaches have been developed, including centralized, hierarchical, and decentralized control of EV charging in offline and online settings. These formulations exploit flexibility in time, space, or both, while optimizing user-, operator-, and/or grid-centric objectives. We refer readers to \cite{elghanam2024optimization} for a systematic review. Notably, the review identifies several \emph{open challenges}, including (1) quantifying the impacts of \emph{spatial} charging coordination on distribution networks, (2) modeling user compliance and participation, and (3) achieving algorithmic scalability for realistic populations and infrastructures.

Regarding \emph{spatial} flexibility and user compliance, Gu \textit{et al.} \cite{gu2025competitive} characterize structural conditions for key properties of a generalized competitive equilibrium (GCE) in power distribution networks with spatially flexible loads, extending classic market theories \cite{hsu1997introduction,caramanis1982optimal}. EV charging behaviors under a GCE can be both feasible and individually incentive-compatible. Chen \textit{et al.} \cite{chen2025quantifying} propose to represent spatio-temporal charging flexibility as virtual batteries and virtual power lines which offers interoperable insights.

Surveys and empirical studies indicate that EV drivers’ willingness, compliance, and satisfaction depend on factors such as incentive schemes, perceived fairness, and service guarantees \cite{hu2025measuring}. Critically, EV drivers prefer to not interrupt their travel schedules from delayed or insufficient charging, which motivates incorporating a mobility perspective \cite{xu2018planning, ju2025trajectory}.

In this regard, we distinguish between \emph{session-based} and \emph{mobility-aware} flexibility models. Session-based scheduling optimizes each charging session in isolation, treating charging demand and dwell time as static and exogenous. It is adopted by the majority of existing works (e.g., \cite{powell2022charging, navidi2023coordinating, wang2025online}), where the parameters are often drawn from oversimplified distributions. In contrast, mobility-aware scheduling jointly accounts for future energy demands and charging opportunities created by a vehicle’s activity schedule. In this case, a charging plan is feasible as long as the vehicle maintains sufficient energy to complete each trip. As illustrated in Fig.~\ref{fig:G-H}\textbf{a}, mobility-aware modeling enlarges the feasible set of charging trajectories (e.g., allowing lower departure SOC from an intermediate stop when downstream charging opportunities exist), thereby unlocking additional flexibility when the grid is constrained.

Finally, \emph{scalability} is essential for high-fidelity modeling but remains challenging. First, individual drivers may charge across a wide set of locations, and charging cycles often span a week or longer \cite{zhan2025large}. Second, both charging infrastructure and distribution capacity are scarce resources in which drivers compete with one another, creating system-level congestion effects that do not arise in single-agent settings.

However, scalability challenges arise from both algorithmic and data perspectives. Algorithmically, long-horizon planning for millions of heterogeneous populations over many alternatives is intractable for many methods that perform well at small scale. Existing scalable approaches include: (1) mean-field models that approximate the population as a continuum (e.g., \cite{ma2011decentralized,tajeddini2018mean}), which gain scalability but often require homogeneity or restrictive heterogeneity assumptions; (2) heuristic aggregation methods \cite{wang2025online,chen2025quantifying} that may not recover globally optimal solutions; and (3) distributed optimization by dual decomposition, in which a coordinator broadcasts a price signal and each EV updates its own charging profile in response \cite{gan2012optimal,le2016optimal,liu2019decentralized}, including its augmented-Lagrangian variant, the alternating direction method of multipliers (ADMM) \cite{rivera2017distributed,zhou2021admm}. Their demonstrations typically involve (i) simplified problem structures, e.g., without network capacity limits \cite{gan2012optimal}; (ii) relatively small scales, e.g., 500 EVs on five feeders \cite{le2016optimal}, whereas at a realistic metropolitan scale, a distributed algorithm requires substantial customization and engineering to work in practice; or (iii) stylized data and scenarios, e.g., battery capacities and states of charge drawn from uniform distributions with a common overnight charging window \cite{liu2019decentralized}.
From a data perspective, realistic insights require fusing multi-modal data sources to reconstruct high-resolution infrastructure constraints and a massive population of heterogeneous EV drivers.%

Finally, we situate the algorithmic route taken in this paper.
Solving a coupled resource-allocation problem by iterating a price against agents' best responses is a classical idea that appears under different names across fields.
In optimization, it is \emph{dual decomposition} solved by the dual gradient method, dating back to Everett \cite{everett1963generalized} and Uzawa \cite{arrow1958studies} and summarized in \cite{boyd2011distributed}.
In communication networks, it underlies network utility maximization, where link prices coordinate distributed users \cite{kelly1998rate, low1999optimization}.
In economics, the price-adjustment reading is known as Walrasian t\^{a}tonnement.
In EV charging, decentralized protocols that iterate a broadcast signal against individual schedule updates follow the same template \cite{gan2012optimal, ma2011decentralized}.
Sec.~\ref{sec:admm} instantiates this lineage as the \emph{iterated price response} (IPR) method.
Relative to the classical scheme, our contributions are an anytime primal--dual certificate at million-EV scale, an entropic price update matched to the geometry of the peak-violation dual, and a shortage-based price initialization.

\subsection{Main Contributions \& Paper Outline}
Relative to the existing literature, our main methodological contributions are as follows:
\begin{enumerate}
\item We characterize mobility-aware charging flexibility by coupling decisions across trips and stays through SOC dynamics, which generalizes conventional session-based constraints. Based on this model, we formulate Mobility-Aware Coordinated EV charging (MAC) as a system-level convex program to mitigate feeder overloading by leveraging spatio-temporal flexibility at scale.

\item We develop the iterated price response (IPR) method, a dual-decomposition algorithm whose iterations are fleet-wide best responses to a locational congestion price, with a convergence guarantee and a practical anytime optimality certificate based on primal--dual bounds. An ADMM-based alternative serves as the algorithmic baseline. To remove per-iteration bottlenecks at population scale, we derive custom primal--dual routines and batched solvers for the dominant EV subproblems.

\item We demonstrate tractability by solving a MAC that coordinates 1.84 million EV charging across 1300+ feeders over a week using high-resolution mobility and grid data.
The unprecedented scale with explicit primal–dual bounds validates the proposed modeling and algorithmic framework, and unveils the potential of regional scale EV charging coordination.
\end{enumerate}

The rest of the paper is outlined as follows (Fig.~\ref{fig:G-H}\textbf{b}):
Section~\ref{sec:mac} introduces the MAC problem by defining mobility-aware charging feasibility and formulating the resulting coordination problem. We also establish a key connection between the social optimum and competitive equilibrium, which motivates practical price-based implementations.
Section~\ref{sec:admm} develops the iterated price response (IPR) method, a dual-decomposition scheme whose iterations are fleet-wide best responses to a locational price, together with its convergence analysis and the primal--dual bounds that certify any candidate solution. A distributed ADMM scheme tailored to MAC's spatial-aggregation structure is derived as the algorithmic baseline.
Section~\ref{sec:customSolver} further derives a custom solver for the per-EV best response, the dominant computational bottleneck of IPR iterations, with per-instance duality-gap tracking. The same machinery covers the subproblems of the ADMM baseline.
Section~\ref{sec:solver_analysis} empirically evaluates convergence behavior and computational performance of the overall solver and its components.
Section~\ref{sec:case-study} reports a regional-scale case study (San Francisco Bay Area) to validate tractability and characterize the system-level effects under mobility-aware coordination.
Finally, the paper concludes in Sec.~\ref{sec:conclusion} with a discussion on limitations and future directions.

\section{EV Charging Coordination Problem}\label{sec:mac}

\subsection{Problem Formulation}
\emph{Notation.}\quad
$\R$ denotes the real numbers, and $[\![N]\!] \coloneqq \{1, \dots, N\}$ for a positive integer $N$.
The operator $[\cdot]^+ \coloneqq \max\{0, \cdot\}$ takes the element-wise positive part.
$\Ind_{\Omega}(\omega)$ is the indicator function of a set $\Omega$, equal to $0$ if $\omega \in \Omega$ and $+\infty$ otherwise.
$\norm{v}_2$ is the Euclidean norm of a vector, and $\norm{M}_F \coloneqq \sqrt{\sum_{i,j} M_{ij}^2}$ and $\langle A, B \rangle_F \coloneqq \sum_{i,j} A_{ij} B_{ij}$ are the Frobenius norm and inner product of matrices.

Let $i \in \Iset$ be the index of an EV, and $t \in \Tset$ be the index of a time slot.
Each EV specifies a charging profile $p_i \in \R^T$, whose entry $p_{i,t}$ is the charging power in time slot $t$. The profile needs to be (individually) \emph{feasible}. For example, it needs to meet the charging energy requirement of mobility, and follow charging power limits.
We denote the feasible set as $\X_i \subset \R^T$, hence $p_i \in \X_i$ for all $i \in \Iset$.

Denote the load profile of distribution feeder $s \in \Sset$ as $l_s \in \R^T$.
Each feeder $s$ serves a certain region. Different EVs move across different service regions, and as a consequence, their charging loads are aggregated onto different feeders at different time slots.
Specifically, given EV itineraries, define a location mapping matrix $\Gamma \in \R^{I \times T}$, where $\Gamma_{i,t} = s$ means that in the slot $t$, the EV $i$ is within the feeder service region $s$.
Define the set $\I_{s,t} \coloneqq \{i \in \Iset: \Gamma_{i,t} = s\}$ as the set of EVs that stop within the region $s$ in the time slot $t$.
The aggregate EV charging load in the region $s$ in the slot $t$ is:
\begin{equation}\label{eq:feeder_agg_st}
    l_{s,t} = \sum_{i\in \I_{s,t}}\, p_{i,t}
\end{equation}

More compactly, define binary routing tensors $A_i \in \{0, 1\}^{S \times T \times T}$ for each EV $i$, where $A_i[s,t,t^\prime] =1$ if and only if $t=t^\prime$ and $\Gamma_{i,t} = s$. Then,
\begin{equation}
    l \in \R^{S\times T} = \sum\nolimits_{i}\, A_i p_i
\end{equation}

Assume the grid-side cost of supplying load $l$ is $\Jgrid(l)$, and each EV $i$ has a disutility function $\Jev_i(p_i)$ for any feasible charging profile $p_i$.
We define the \textbf{social optimum (SO)} of coordinated EV charging as the solution to \eqref{eq:opt_SO}:
\begin{subequations}\label{eq:opt_SO}
\begin{align}
    ({\rm SO})\qquad \min_{p, l}~~ & \Jgrid(l) + \sum\nolimits_i \Jev_i(p_i)\\
    \text{s.t.}~~ & p_i \in \X_i, \quad \forall i\label{eq:EV_i_C0}\\
    & \sum\nolimits_{i}\, A_i p_i - l = 0\quad (\text{dual:}~\lambda)
\end{align}
\end{subequations}

On the other hand, consider a market in which a locational price $\lambda$ is posted.
Each of the $I$ EV owners minimizes their disutility $\Jev_i$ plus the charging payment under this price, subject to private constraints $\mathcal{X}_i$, while a grid operator chooses the aggregate load $l$ it serves to maximize its revenue minus the supply cost $\Jgrid$.
A \textbf{competitive equilibrium (CE)} is a price-allocation pair $(p^\#, l^\#, \lambda^\#)$ that satisfies \eqref{eq:opt_CE}:
\begin{subequations}\label{eq:opt_CE}
\begin{align}
    ({\rm CE})\quad & p^\#_i \in \arg\min_{p_i\in \X_i}~\Jev_i(p_i) + \langle A_i^\top \lambda^\#, p_i\rangle, ~ \forall i\label{eq:EV_i_CE}\\
    & l^\# \in \arg\max_{l\succeq 0} ~ -\Jgrid(l) + \langle\lambda^\#, l\rangle\label{eq:l_CE}\\
    & l^\# = \sum\nolimits_i A_i p^\#_i \label{eq:clear_CE}
\end{align}
\end{subequations}

\begin{theorem}\label{thm:SO_CE}
Assume that (A1) $\X_i$'s are nonempty, closed, and convex; (A2) $\Jev_i$'s and $\Jgrid$ are proper, closed, and convex;
(A3) strong duality holds for (SO), and dual optimum is attained. Then:
$(p^\star, l^\star, \lambda^\star)$ is a primal-dual optimal solution (i.e., a KKT solution) to (SO) \textit{iff}
$(p^\star, l^\star, \lambda^\star)$ is a CE.
\end{theorem}

\begin{theorem}\label{thm:SO_CE-2}
Assume (A1)--(A3). In addition, assume (A4) $\Jev$'s are \emph{strongly} convex.
Let $\lambda^\#$ be any dual optimal multiplier of (SO),
$p^\#_i$ be the best response under $\lambda^\#$, i.e. a solution to \eqref{eq:EV_i_CE},
and $l^\# = \sum_i A_i p^\#_i$. Then:
$(p^\#, l^\#, \lambda^\#)$ is a CE. Therefore, by Thm.\ref{thm:SO_CE}, $(p^\#, l^\#)$ is socially optimal.
\end{theorem}

Thm.\ref{thm:SO_CE} and Thm.\ref{thm:SO_CE-2} state that,
on the one hand, the social optimum is implementable with a dynamic locational price $\lambda^\#$,
and on the other hand, such $\lambda^\#$ can be obtained if we can find a KKT solution to \hyperref[eq:opt_SO]{(SO)} effectively.

\subsection{Mobility-aware Coordination to Minimize Overloading}

In this paper, we introduce the notion of \textbf{mobility-aware} EV charging feasibility.
Unlike conventional notions that define session-specific energy targets,
mobility-aware feasibility only requires maintaining proper energy levels throughout the time horizon (typically a week) such that the EV has enough mileage to complete each trip.
Formally,
\begin{definition}[Mobility-aware feasibility]\label{defn:feas}
The mobility-aware feasibility set $\X_i \subset \R^T$ of EV $i$ is the collection of charging profiles $p_i$ that verify the constraints\footnote{Note that using $e_{i,t} = E_{i,0} + \sum_{\tau\le t} (p_{i, \tau}\Delta t-D_{i,\tau})$, we can rewrite the constraints in \eqref{eq:EV_i} only with the decision variables $(p_{i,t})$.}:
\begin{subequations}\label{eq:EV_i}
\begin{align}
    \forall t \in \Tset: \quad
    & e_{i,t} = e_{i,t-1} + p_{i,t}\Delta t - D_{i,t} \\
    & \underline{P}_{i,t} \le p_{i,t} \le \overline{P}_{i,t}\\
    & \underline{E}_{i, t} \le e_{i,t} \le \overline{E}_{i, t}\label{eq:ev_constr_Elim}\\
    \text{and}\qquad& e_{i,0} = E_{i,0} %
\end{align}
\end{subequations}
where $p_{i,t}$ denotes the charging power of EV $i$ in slot $t$, $\Delta t$ is the slot length,
$e_{i,t}$ denotes the battery energy level (state of energy) of EV $i$ at the end of slot $t$, and $D_{i,t}$ is its driving energy consumption in slot $t$, which can be determined by an energy consumption model given the itinerary of EV $i$.
\end{definition}

\begin{remark}
Time-varying maximum (minimum) charging power $\overline{P}_{i,t}$ ($\underline{P}_{i,t}$) provides flexibility in modeling charger access at different stays.
For example, if EV $i$ is on a trip throughout time slot $t$, or there are no nearby chargers around its stay location, then it cannot be charged, so $\overline{P}_{i,t}=0$.
We also let the minimum (maximum) battery energy level $\underline{E}_{i,t}$ ($\overline{E}_{i,t}$) vary over time.
By default, $\underline{E}_{i,t} = 0$ (and $\overline{E}_{i,t} = E^{\max}_i$, which is the battery capacity).
We may force additional constraints on energy requirements for specific stays. By default, we set $\underline{E}_{i,T} = E_{i, 0}$, which requires the terminal energy level to be no less than the initial level and thus ensures weekly energy balance.
Setting energy targets at additional departure slots recovers the conventional session-specific energy targets as a special case.
\end{remark}

In this paper we are primarily interested in understanding how coordinated EV charging could help mitigate distribution feeder overloading.
Each feeder $s$ has a time-varying \emph{hosting capacity}, $C_{s,t}$, which denotes the maximum demand allowed to add on top of its base load without violating thermal, voltage, and safety constraints.
We define load violation on feeder $s$ as the maximum load surplus over its host capacity over time, i.e.,
\begin{equation}
    v_s(l_s) = \max\nolimits_t\, \left[l_{s,t} - C_{s,t}\right]^{+}
\end{equation}
and set the grid side goal to minimize the total violations, namely,
\begin{equation}\label{eq:Jgrid}
    \Jgrid(l) = \sum\nolimits_s v_s(l_s)
\end{equation}

For the disutility function of individual EVs, we consider a simple regularization term:
\begin{equation}\label{eq:Jev}
    \Jev_i(p_i) = \frac{\kappa}{2} \norm{p_i}_2^2
\end{equation}
where parameter $\kappa > 0$ is a small constant, which can be viewed as a battery degradation cost \cite{le2016optimal}.

\begin{remark}
It is obvious to see that (i) $\X_i$'s are convex, and (ii) $\Jgrid$ is convex and $\Jev_i$'s are $\kappa$-strongly convex.
Moreover, (iii) provided that each $\X_i$ is nonempty, strong duality holds for (SO) and the dual optimal value is attained, because all constraints of (SO) are polyhedral and the refined Slater's condition reduces to feasibility.
Therefore, the equivalence of SO and CE stated in Thm.\ref{thm:SO_CE} and Thm.\ref{thm:SO_CE-2} applies here.
\end{remark}

\section{The Iterated Price Response Method}\label{sec:admm}

If \hyperref[eq:opt_SO]{(SO)} is solved as a whole, then the number of constraints and variables is $\mathcal{O}(IT)$ (assume $I >> S$), which can be prohibitively large for the scope of the analysis we are interested in.
For example, the number of EVs, $I$, might be of the order of one million.
Meanwhile, the constraint sets $\X_i$ are uncoupled across EVs. The EVs interact only through the aggregate feeder loads.
Leveraging this structural property, this section develops a scalable dual decomposition scheme solved by the dual gradient method.
The approach maximizes the Lagrangian dual of \hyperref[eq:opt_SO]{(SO)} directly by first-order ascent.
We call the coordination pipeline \emph{iterated price response} (IPR). Each iteration posts a locational price, lets every EV best-respond in parallel, and updates the price from the observed capacity shortage.
The general dual decomposition scheme has a long history in the optimization literature \cite{everett1963generalized, boyd2011distributed}, and is also known as Uzawa's method \cite{arrow1958studies}.

The rest of this section unfolds as follows.
Sec.~\ref{sec:ipr_alg} derives the IPR scheme via dual decomposition and shows that the dual gradient is exactly the capacity shortage deployed by fleet-wide best responses, which also furnishes an anytime optimality certificate.
Sec.~\ref{sec:ipr_conv} establishes convergence guarantees of IPR.
Sec.~\ref{sec:ipr_entropic} instantiates the ascent as an entropic mirror scheme with a shortage-based initialization, which improves the convergence.
Sec.~\ref{sec:admm_alt} derives ADMM, another classical distributed optimization scheme, as the algorithmic baseline for comparison.

\subsection{Dual Decomposition and Price-Response Gradients}\label{sec:ipr_alg}

By relaxing the coupling constraint of \hyperref[eq:opt_SO]{(SO)} with a dual variable $\lambda \in \R^{S\times T}$, we obtain the Lagrangian:
\begin{align}\label{eq:SO_Lagrangian}
    \mathcal{L}(p, l, \lambda) = & \sum\nolimits_i \Jev_i(p_i) + \Jgrid(l) + \langle \lambda, \sum\nolimits_i A_i p_i - l\rangle
\end{align}
The corresponding dual function $Q(\lambda) \coloneqq \min \{\mathcal{L}(p, l, \lambda): p_i \in \X_i~\forall i,\, l \in \R^{S\times T}\}$ is separable as:
\begin{align}
    &Q(\lambda) = \sum\nolimits_i\text{(D1)} + \text{(D2)}, ~~\text{where}\notag\\
    \text{(D1)}~~ &\min_{p_i \in \X_i}~\Jev_i(p_i) + \langle A_i^\top \lambda, p_i \rangle\label{eq:D1}\\
    \text{(D2)}~~ &\min_l~\Jgrid(l) - \langle \lambda, l\rangle \label{eq:D2}
\end{align}
Note that \hyperref[eq:D1]{(D1)} and \hyperref[eq:D2]{(D2)} are exactly the best-response problems \eqref{eq:EV_i_CE} and \eqref{eq:l_CE} of the competitive equilibrium, respectively.

\begin{lemma}\label{lem:D2}
With $\Jgrid$ as in \eqref{eq:Jgrid}, the (D2) term admits a closed form.
\begin{equation}
\text{\rm (D2)} =
\begin{cases}
    -\langle \lambda, C \rangle, & \lambda \in \Lambda\\
    -\infty, & \text{otherwise}
\end{cases}
\end{equation}
where the \emph{dual-feasible set} $\Lambda$ is
\begin{equation}\label{eq:Lambda}
    \Lambda \coloneqq \left\{\lambda \in \R^{S \times T}: \lambda_s \succeq 0,~ \textstyle\sum\nolimits_t \lambda_{s,t} \le 1, ~\forall s\right\}
\end{equation}
\end{lemma}
\begin{proof}
See Appendix~\ref{appx:S2_solver}.
\end{proof}

The set $\Lambda$ is a product of $S$ \emph{capped simplices}, with each feeder carrying at most one unit of price mass over its slots.
This is a direct consequence of $\Jgrid$ being a sum of per-feeder maxima, since the dual price of a peak constraint is a probability-like weighting over the slots that attain the peak.
Evaluating the (D2) term is therefore cheap. One checks $\lambda_s \succeq 0$ and $\sum_t \lambda_{s,t} \le 1$ and returns $-\langle \lambda_s, C_s\rangle$ or $-\infty$ accordingly, at a cost of $\mathcal{O}(T)$ per feeder.

By Lemma~\ref{lem:D2}, we restrict the dual variable to $\lambda \in \Lambda$, on which
\begin{equation}\label{eq:Q_on_Lambda}
    Q(\lambda) = \sum\nolimits_i \xi_i(\lambda) - \langle \lambda, C\rangle
\end{equation}
where $\xi_i(\lambda)$ denotes the optimal value of \hyperref[eq:D1]{(D1)}.
Since $\kappa > 0$, the objective of \hyperref[eq:D1]{(D1)} is strongly convex, so its minimizer is unique. Denote it by $p_i(\lambda)$, the \emph{best response} of EV $i$ to the price $\lambda$.
By Danskin's theorem \cite[Prop.~B.25]{bertsekas1999nonlinear}, $\xi_i$ is differentiable with $\nabla \xi_i(\lambda) = A_i\, p_i(\lambda)$. Therefore,
\begin{equation}\label{eq:gradQ}
    \nabla Q(\lambda) = \sum\nolimits_i A_i\, p_i(\lambda) - C
\end{equation}
In other words, \emph{the gradient of the dual is the deployed capacity shortage}. Evaluating it amounts to posting the price and aggregating the fleet's responses, which is precisely one round of price response.

The basic method is projected gradient ascent on $Q$ over $\Lambda$:
\begin{equation}\label{eq:ipr_generic}
    \lambda^{(k+1)} = \Pi_\Lambda\!\left[\lambda^{(k)} + \eta\, \nabla Q(\lambda^{(k)})\right]
\end{equation}
where $\Pi_\Lambda$ denotes the projection onto $\Lambda$ and $\eta > 0$ is a step size. Sec.~\ref{sec:ipr_entropic} instantiates the ascent in a non-Euclidean geometry better matched to $\Lambda$.

\smallskip
\noindent\textbf{Anytime optimality certificate.}\quad
Weak duality turns the dual value into a certificate of solution quality.
For any feasible $\widetilde{p}$ (i.e., $\widetilde{p}_i \in \X_i$ for all $i$), let $J(\widetilde{p}) \coloneqq \sum\nolimits_i \Jev_i(\widetilde{p}_i) + \Jgrid(\sum\nolimits_i A_i \widetilde{p}_i)$ denote its primal objective, and let $J^\star$ be the optimal objective of \hyperref[eq:opt_SO]{(SO)}. Then
\begin{equation}\label{eq:weak_duality}
Q(\lambda) \le J^\star \le J(\widetilde{p}), \quad
\forall \widetilde{p} \in \Pi_i \X_i, ~\forall \lambda \in \Lambda
\end{equation}
so the optimality gap of $\widetilde{p}$ is bounded by $J(\widetilde{p}) - Q(\lambda)$, and a practical stopping criterion is the relative version
\begin{equation}\label{eq:SO_rel_gap}
    J(\widetilde{p}) - Q(\lambda) \le \epsilon^{\rm rel} \cdot
    \max\{1, |J(\widetilde{p})|, |Q(\lambda)|\}
\end{equation}
where $\epsilon^{\rm rel}>0$ is a given relative optimality gap tolerance.
Every iteration of the ascent \eqref{eq:ipr_generic} natively produces both sides of \eqref{eq:SO_rel_gap}.
The responses give a feasible candidate $p(\lambda^{(k)})$, hence the upper bound $J$.
The dual iterate $\lambda^{(k)}$ gives the lower bound $Q(\lambda^{(k)})$, whose (D1) terms are already evaluated in the course of computing the best responses.
We track the best certified pair along the ascent and report it upon stopping, so the method is \emph{anytime}. Terminating at any iteration returns a feasible solution together with a proven optimality gap.
The certificate remains valid even when a practical sub-solver returns inexact or infeasible responses $p(\lambda)$, and we defer this technical treatment to the subproblem solvers of Sec.~\ref{sec:customSolver}.

\subsection{Convergence Analysis}\label{sec:ipr_conv}

Define $N_{s,t} \coloneqq |\I_{s,t}|$ as the number of EVs within region $s$ in slot $t$, and $N_{\max} \coloneqq \max_{s,t} N_{s,t}$.

\begin{lemma}[Smoothness of the dual]\label{lem:Q_smooth}
Suppose $\kappa>0$ and $\X_i \neq \emptyset$ for all $i$. Then $Q$ is concave and continuously differentiable on $\Lambda$ with gradient \eqref{eq:gradQ}, and $\nabla Q$ is Lipschitz continuous with constant $L_Q \le N_{\max} / \kappa$, i.e.,
\begin{equation}
\norm{\nabla Q(\lambda) - \nabla Q(\lambda^\prime)}_F \le \frac{N_{\max}}{\kappa} \norm{\lambda - \lambda^\prime}_F, \quad \forall \lambda, \lambda^\prime \in \Lambda
\end{equation}
\end{lemma}
\begin{proof}
See Appendix~\ref{appx:ipr}.
\end{proof}

The strong convexity from $\kappa > 0$ is what makes $Q$ smooth. With $\kappa = 0$, \hyperref[eq:D1]{(D1)} is a linear program whose minimizer is set-valued, and $Q$ is only piecewise linear.
Adding a strongly convex term to smooth a non-smooth function is a classical technique \cite{nesterov2005smooth}.
Here, the disutility term is part of the model rather than an added approximation, so no smoothing error needs to be removed afterwards.
The same term is also assumption (A4) of Thm.~\ref{thm:SO_CE-2}, which makes the optimal price deployable.

\begin{theorem}[Convergence of IPR]\label{thm:ipr_conv}
Suppose the assumptions of Lemma~\ref{lem:Q_smooth} hold, so that strong duality holds for \hyperref[eq:opt_SO]{(SO)} (cf.\ the Remark in Sec.~\ref{sec:mac}).
Let $\lambda^\star \in \Lambda$ be a dual optimal solution and $p^\star$ the unique optimal charging profile of \hyperref[eq:opt_SO]{(SO)}.
Under the update \eqref{eq:ipr_generic} with $\eta = 1/L_Q$:
\begin{subequations}
\begin{enumerate}
\item[(a)] the dual values converge at rate $\mathcal{O}(1/k)$:
\begin{equation}\label{eq:ipr_conv_dual}
Q(\lambda^\star) - Q(\lambda^{(k)}) \le \frac{L_Q \norm{\lambda^{(0)} - \lambda^\star}_F^2}{2k}
\end{equation}
\item[(b)] for every $k$, the deployed best response satisfies
\begin{equation}\label{eq:ipr_conv_primal}
\frac{\kappa}{2} \sum\nolimits_i \norm{p_i(\lambda^{(k)}) - p_i^\star}_2^2 \le Q(\lambda^\star) - Q(\lambda^{(k)})
\end{equation}
hence it converges to the social optimum at rate $\mathcal{O}(1/\sqrt{k})$.
\end{enumerate}
\end{subequations}
\end{theorem}
\begin{proof}
See Appendix~\ref{appx:ipr}.
\end{proof}

Part (b) concerns the \emph{price-response allocation} $p(\lambda^{(k)})$, which is what a price-based implementation realizes.
What ultimately matters for deployment is its objective $J(p(\lambda^{(k)}))$.
Part (b) shows the allocation converges to the optimum, so $J(p(\lambda^{(k)})) \to J^\star$ by continuity.

\begin{remark}\label{rmk:ipr_step}
The Lipschitz constant $L_Q$ is large in practice ($N_{\max}$ reaches the thousands at scale while $\kappa$ is small), so the constant step $\eta = 1/L_Q$ is conservative and the worst-case bound \eqref{eq:ipr_conv_dual} is pessimistic.
The implementation therefore uses normalized, backtracked steps (Sec.~\ref{sec:ipr_entropic}) that enforce monotone ascent of $Q$, under which the $\mathcal{O}(1/k)$ guarantee is retained while practical progress is much faster.
\end{remark}

\begin{algorithm}
\caption{Iterated price response (IPR)}\label{alg:ipr}
\begin{algorithmic}[1]\small
\Procedure{IPR}{$C$, $\beta$, $\eta_0$, $K$}
\State $p_i \gets \arg\min_{p_i \in \X_i} \Jev_i(p_i)$, $\forall i$
\Comment{unpriced sweep}
\State $\lambda \gets \widehat{\lambda}(\sum_i A_i p_i)$ by \eqref{eq:lse_price}
\Comment{shortage seed}
\State $\eta \gets \eta_0$
\For{$k = 0, 1, \dots, K-1$}
\State $p_i(\lambda) \gets$ solve \hyperref[eq:D1]{(D1)}, $\forall i$, by Alg.~\ref{alg:step1}
\Comment{best responses}
\State $g \gets \sum_i A_i\, p_i(\lambda) - C$
\Comment{$= \nabla Q(\lambda)$}
\State evaluate $J$, $Q$; keep the best certified pair
\State $\lambda_s^+ \gets \Pi_\Lambda\!\left[\lambda_s \odot \exp(\eta\, g_s / c_s)\right]$, $\forall s$, by \eqref{eq:polish}
\State backtrack $\eta \gets \eta/2$ until $Q(\lambda^+) \ge Q(\lambda)$; $\lambda \gets \lambda^+$
\EndFor
\State \Return best $(p(\lambda), \lambda)$, certified gap by \eqref{eq:SO_rel_gap}
\EndProcedure
\end{algorithmic}
\end{algorithm}

\subsection{Entropic Mirror Ascent and Shortage-Based Initialization}\label{sec:ipr_entropic}

The Euclidean update \eqref{eq:ipr_generic} treats all coordinates of $\lambda$ alike, whereas the shortages that drive the gradient differ by orders of magnitude across feeders.
We therefore adopt \emph{entropic mirror ascent} \cite[Ch.~9]{beck2017first}, whose update is multiplicative per feeder:
\begin{equation}\label{eq:polish}
    \lambda_s \gets \Pi_{\Lambda}\!\left[\lambda_s \odot \exp\!\big(\eta\, g_s / c_s\big)\right], \qquad g = \nabla Q(\lambda)
\end{equation}
where $\odot$ and $\exp$ apply element-wise, the scale $c_s = \max\{\norm{g_s}_\infty, c_0\}$ (with a small floor $c_0>0$) normalizes the step per feeder, and $\Pi_\Lambda$ here rescales each feeder's price mass into the capped simplex.
The step size $\eta$ is backtracked. A candidate is accepted only if $Q$ does not decrease, so the ascent is monotone and the guarantee of Thm.~\ref{thm:ipr_conv} is preserved.

Three properties motivate the multiplicative form.
First, it is scale-free, since one step size serves feeders whose shortages differ by orders of magnitude, with no per-feeder tuning.
Second, non-negativity of the price is automatic.
Third, price mass concentrates exponentially on persistently violated slots, matching the structure of the optimal price, which is supported on the slots that attain each congested feeder's peak.
Two safeguards address the known weakness that zero is absorbing under multiplicative updates. The projection $\Pi_\Lambda$ floors tiny entries, and it seeds a small price mass on violated feeders that carry none, so that support can re-enter.
Both safeguards are accepted under the same monotone-$Q$ rule.

A good initial price follows from the same object the gradient exposes.
Given a reference aggregate load $\hat{l}$, e.g., from the unpriced responses $\arg\min_{p_i \in \X_i} \Jev_i(p_i)$, let $\delta_s = \hat{l}_s - C_s$ be the shortage profile of feeder $s$ and $\overline{\delta}_s = \max_t \delta_{s,t}$.
For every congested feeder ($\overline{\delta}_s > 0$), define the smoothed worst-case price
\begin{equation}\label{eq:lse_price}
    \widehat{\lambda}_{s,t}(\hat{l}) = \frac{\exp(\beta\, \delta_{s,t} / \overline{\delta}_s)}{\sum_{\tau}\exp(\beta\, \delta_{s,\tau} / \overline{\delta}_s)}
\end{equation}
with a sharpness parameter $\beta > 0$, while uncongested feeders receive $\widehat{\lambda}_s = 0$.
Each $\widehat{\lambda}_s$ is a smoothed subgradient of $\max_t\,[\cdot]^+$ and lies in $\Lambda$, with $\beta \to \infty$ recovering the argmax vertex.
The smoothing is motivated by deployment. A price concentrated on a single slot would herd the price-responding EVs into that slot, whereas the dense form prices every slot of a congested feeder.
Initializing at $\lambda^{(0)} = \widehat{\lambda}(\hat{l})$ starts the ascent at the natural magnitude of the congestion price, whereas ascending from $\lambda = 0$ spends many early iterations recovering that magnitude.

Alg.~\ref{alg:ipr} summarizes the method.
Each iteration costs one \hyperref[eq:D1]{(D1)} sweep over the fleet (Sec.~\ref{sec:solver_S1}), fully parallel across EVs.

\subsection{Comparison with the ADMM Alternative}\label{sec:admm_alt}

The alternating direction method of multipliers (ADMM) is an alternative scheme for distributed optimization that treats \hyperref[eq:opt_SO]{(SO)} by operator splitting  \cite{boyd2011distributed}.
By introducing auxiliary variables $z_i \in \R^{S \times T} = A_i p_i$, we can rewrite the problem \hyperref[eq:opt_SO]{(SO)} as:
\begin{subequations}\label{eq:ADMM_0}
\begin{align}
    ({\rm P})\qquad
    \min_{p, z}~ & \sum\nolimits_{i}\, f_i(p_i) + g\left(\sum\nolimits_i z_i\right)\\
    \text{s.t.}\quad & z_i =  A_i p_i \quad (\text{dual}:~\lambda_i),\quad \forall i \label{eq:couple_constr}
\end{align}
\end{subequations}
where
\begin{align*}
    f_i(p_i):~ & \R^T \mapsto \R \cup \{+\infty\} \coloneqq \Ind_{\X_i}(p_i) + \Jev_i(p_i)\\
    g(l): ~& \R^{S\times T} \mapsto \R \coloneqq \Jgrid(l)
\end{align*}

From \hyperref[eq:ADMM_0]{(P)}, ADMM constructs an augmented Lagrangian and alternates minimization over $p$ and $z$ with a multiplier update.
Define $\lambda_i \in \R^{S \times T}$ to be the dual variable corresponding to the $i^{\rm th}$ coupling constraint \eqref{eq:couple_constr}, let $r_i \in \R^{S \times T} \coloneqq A_i p_i - z_i$, and let $\rho>0$ be a \emph{penalty} factor, set to $\rho = 0.1$ in all ADMM runs of this paper.
With the scaled dual $\mu \coloneqq \lambda / \rho$, the (scaled) augmented Lagrangian is:
\begin{align}\label{eq:Lag_scale}
    L_\rho(p, z, \mu) & \coloneqq \sum\nolimits_{i}\, f_i(p_i) + g(\sum\nolimits_i z_i) \notag\\
    & \quad + \frac{\rho}{2}\sum\nolimits_i \left(\norm{r_i+\mu_i}_F^2 - \norm{\mu_i}_F^2\right)
\end{align}

In each iteration $k$, ADMM performs a minimization of $L_\rho$ over $p$ (step S1), a minimization over $z$ (step S2), and a \emph{multiplier update} $\mu^{(k+1)} \gets \mu^{(k)} + r^{(k+1)}$ (step S3).
Problem \hyperref[eq:ADMM_0]{(P)} further belongs to the class of \emph{sharing} problems \cite[Sec.~7.3]{boyd2011distributed}. By exploiting this structure, the $z$ update admits an analytical solution and the per-EV duals $\mu_i$ collapse to their common average $\mu$ (consensus).
The three steps then reduce to (see Appendix~\ref{appx:admm} for the derivation):
\begin{subequations}\label{eq:admm_steps}
\begin{align}
&\text{(S1)} && p_i^{(k+1)} \gets \arg\min_{p_i \in \X_i} \frac{\rho}{2}\norm{A_i p_i - b_i^{(k)}}_F^2 + \Jev_i(p_i)\label{eq:step_1}\\
&\text{(S2)} && l^{(k+1)} \gets \arg\min_l \frac{\rho}{2I}\norm{l - D^{(k)}}^2_F + \Jgrid(l)\label{eq:step_2}\\
&\text{(S3)} && \mu^{(k+1)} \gets \mu^{(k)} + \frac{1}{I} (\hat{l}^{(k+1)} - l^{(k+1)})\label{eq:step_3}
\end{align}
\end{subequations}
with $\hat{l}^{(k)} = \sum_i A_i p_i^{(k)}$, $b_i^{(k)} = A_i p_i^{(k)} - \frac{1}{I}(\hat{l}^{(k)} - l^{(k)} + I \mu^{(k)})$,
and $D^{(k)} = \hat{l}^{(k+1)} + I \mu^{(k)}$.
Step \hyperref[eq:step_1]{(S1)} solves $I$ independent EV subproblems in parallel, and \hyperref[eq:step_2]{(S2)} decomposes into $S$ independent feeder subproblems.
(S1) is solved by an adaptation of the custom solver of Sec.~\ref{sec:customSolver}, and (S2) admits an exact $\mathcal{O}(T \ln T)$ per-feeder solution (see Appendix~\ref{appx:S2_solver}).

The general convergence results \cite[Sec.3.2]{boyd2011distributed} apply here:
\begin{theorem}\label{thm:admm_conv}
As $k \to \infty$,
\begin{itemize}
    \item the primal residual $r^{(k)} \to 0$;
    \item the objective $\sum\nolimits_{i}\, f_i(p^{(k)}_i) + g(\sum\nolimits_i z_i^{(k)}) \to J^\star$;
    \item the dual variable $\lambda^{(k)} \to \lambda^\star$, where $\lambda^\star$ is a dual optimal point of \hyperref[eq:ADMM_0]{(P)}
\end{itemize}
\end{theorem}
\begin{proof}
The required assumptions are verified in Appendix~\ref{appx:admm}.
\end{proof}

Two observations position the two methods.
First, Thm.~\ref{thm:admm_conv} requires no strong convexity. It holds even with $\kappa = 0$, in which case the IPR analysis of Sec.~\ref{sec:ipr_conv} fails ($Q$ becomes non-smooth and best responses set-valued).
This robustness is ADMM's genuine theoretical advantage, although a \emph{deployable} price requires $\kappa > 0$ either way (Thm.~\ref{thm:SO_CE-2}).
Second, ADMM's guarantees concern objective values, residuals, and the limit of the dual variable, none of which controls the deployed shortage $\nabla Q(\rho\mu^{(k)})$ at the \emph{current} price, which by \eqref{eq:gradQ} is what a price-based deployment realizes.
The deployed response still converges to $p^\star$ in the limit, since $\lambda^{(k)} \to \lambda^\star$, the best response is Lipschitz continuous in the price \eqref{eq:BR_lipschitz}, and $p(\lambda^\star) = p^\star$ by Thm.~\ref{thm:ipr_conv}(b), but no rate follows.
Moreover, the multiplier update \eqref{eq:step_3} advances the price by the average per-EV residual, carrying a factor $1/I$, so at population scale the price magnitude builds slowly from a cold start.
IPR sidesteps both issues by ascending $Q$ directly from the shortage-based initialization \eqref{eq:lse_price}.
Sec.~\ref{sec:solver_analysis} compares the two methods empirically, and accelerated variants of the ADMM baseline are discussed in Appendix~\ref{appx:accel}.

\section{Custom Subproblem Solvers}\label{sec:customSolver}
Each IPR iteration solves $I$ instances of the best response \hyperref[eq:D1]{(D1)}, one per EV.
At population scale, this sweep dominates the runtime of the whole pipeline.
The instances are readily solved by generic off-the-shelf solvers, but with millions of them per iteration, per-instance speed becomes decisive.
This section develops a custom solver that exploits the structure of $\X_i$ to accelerate \hyperref[eq:D1]{(D1)} while preserving certifiable primal--dual bounds.
The subproblems of the ADMM baseline are covered by the same machinery (see the remark closing Sec.~\ref{sec:solver_S1}).

\subsection{Structure of the Feasible Set $\X_i$}\label{sec:feas}
For $t \in \Tset$, define
\begin{equation}\label{eq:S_bounds}
\begin{aligned}
& \underline{S}_{i,t} = \frac{1}{\Delta t}\Big(\underline{E}_{i,t} - E_{i,0} + \sum\nolimits_{\tau\le t} D_{i,\tau}\Big)\\
& \overline{S}_{i,t} = \frac{1}{\Delta t}\Big(\overline{E}_{i,t} - E_{i,0} + \sum\nolimits_{\tau\le t} D_{i,\tau}\Big)
\end{aligned}
\end{equation}
then the mobility-aware feasibility set $\X_i$ in Defn.\ref{defn:feas} can be rewritten as:
\begin{align}
\X_i  = \{p_i\in \R^T: \underline{P}_i \preceq p_i \preceq \overline{P}_i, \ \underline{S}_i \preceq Hp_i \preceq \overline{S}_i\}
\end{align}
where $H \in \R^{T \times T}$ is the cumulative-sum operator, i.e. $H_{jk} = \mathbf{1}\{j \ge k\}$, with $\mathbf{1}\{A\}$ equal to $1$ if statement $A$ is true and $0$ otherwise, so $(Hp_i)_t = \sum_{\tau\le t} p_\tau$.
Clearly, $\X_i$ is a closed polytope.

We characterize an equivalent condition for $\X_i \neq \emptyset$, which can be verified in $\mathcal{O}(T)$.
Define $\underline{C}_i \coloneqq H \underline{P}_i, \overline{C}_i \coloneqq H \overline{P}_i$.
Obviously, it is \emph{necessary} for $\X_i \neq \emptyset$ to ensure that
\begin{equation}\label{eq:feas_necc}
\underline{P}_i \preceq \overline{P}_i, ~\text{and}~
\underline{C}_i \preceq \underline{S}_i \preceq \overline{S}_i \preceq \overline{C}_i
\end{equation}
However, \eqref{eq:feas_necc} alone is not \emph{sufficient}.
Define $\underline{B}_i \in \R^T \coloneqq \overline{C}_i + \underline{G}_i$, $\overline{B}_i \in \R^T \coloneqq \underline{C}_i + \overline{G}_i$, where
$$
    \underline{G}_{i,t} \coloneqq \max_{\tau \ge t}\,\{\underline{S}_{i,\tau} - \overline{C}_{i,\tau}\}, ~~
    \overline{G}_{i,t} \coloneqq \min_{\tau \ge t}\,\{\overline{S}_{i,\tau} - \underline{C}_{i,\tau}\}
$$

\begin{theorem}[Feasibility condition]\label{thm:feas_oracle}
For EV $i$, assume that \eqref{eq:feas_necc} satisfies, then:
$\X_i \neq \emptyset$ iff. $\underline{B}_i \preceq \overline{B}_i$.

This implies that one can check the feasibility in $\mathcal{O}(T)$, which is the time complexity for the construction and comparison of $\underline{B}_i, \overline{B}_i$.
\end{theorem}

\begin{proof}
See Appendix~\ref{appx:feas_oracle}.
\end{proof}

Moreover, provided an arbitrary charging profile $\hat{p}_i \in \R^T$ that might be infeasible,
Alg.\ref{alg:feas_oracle} is a $\mathcal{O}(T)$ greedy repairing algorithm that returns a feasible $p_i \in \X_i$ that is ``close'' to $\hat{p}_i$, denoted as
$p_i = \verb|Feas|(\hat{p}_i)$
(or return \verb|infeas| if $\X_i = \emptyset$).

\begin{algorithm}
\caption{Find a feasible $p_i$ close to a target $\hat{p}_i \in \R^T$}\label{alg:feas_oracle}
\begin{algorithmic}[1]\small
\Procedure{Feas}{$\hat{p}_i$}
\State \Comment{Construction of $\underline{B}_i, \overline{B}_i$ is $\hat{p}_i$-independent}
\If{$\underline{B}_{i} \npreceq \overline{B}_{i}$}
\State \Return \verb|infeas|
\EndIf
\State $S_{i,0} \gets 0$
\For{$t = 1, ..., T$}
\State $\underline{\beta}_{i,t} \gets \max\{\underline{B}_{i,t}, S_{i,t-1}+\underline{P}_{i,t}\}$
\State $\overline{\beta}_{i,t} \gets \min\{\overline{B}_{i,t}, S_{i,t-1}+\overline{P}_{i,t}\}$
\State $S_{i,t} \gets \text{clip}(S_{i,t-1} + \hat{p}_{i,t}; \underline{\beta}_{i,t}, \overline{\beta}_{i,t})$
\State $p_{i,t} \gets S_{i,t}- S_{i,t-1}$
\EndFor
\State \Return  $p_i$
\EndProcedure
\end{algorithmic}
\end{algorithm}

Alg.\ref{alg:feas_oracle} turns out to be important for the implementation of a robust and certifiable subproblem solver in Sec.\ref{sec:solver_S1}.

\subsection{A Certifiable Solver for the Best Response (D1)}\label{sec:solver_S1}

Subproblem \hyperref[eq:D1]{(D1)}, the gradient oracle of IPR, minimizes a strongly convex quadratic over the box and cumulative-energy constraints of $\X_i$.
Concretely, with $\Jev_i(p_i) = \frac{\kappa}{2}\norm{p_i}_2^2$ and the description of $\X_i$ from Sec.~\ref{sec:feas}, \hyperref[eq:D1]{(D1)} seeks to solve the following optimization problem:
\begin{subequations}\label{eq:D1_exp}
\begin{align}
\min_{p_i}~~& \frac{\kappa}{2}\norm{p_i}_2^2 + \langle c_i, p_i \rangle\\
\text{s.t.}~~& \underline{P}_i \preceq p_i \preceq \overline{P}_i\\
& \underline{S}_i \preceq Hp_i \preceq \overline{S}_i
\end{align}
\end{subequations}
where $c_i \in \R^T \coloneqq A_i^\top \lambda$ collects the prices EV $i$ is exposed to along its itinerary.

Define $\underline{u}_i, \overline{u}_i \in \R^T$ as dual variables corresponding to constraints $\underline{S}_i \preceq Hp_i$ and $Hp_i \preceq \overline{S}_i$, respectively,
and let $u_i \in \R^{2T} \coloneqq [\underline{u}_i, \overline{u}_i]$.
Define Lagrangian $\mathcal{L}(p_i, u_i)$ as
\begin{align}
\mathcal{L}(p_i, u_i) =
    & \frac{\kappa}{2}\norm{p_i}_2^2 + \langle c_i, p_i \rangle \notag\\
    & + \langle \underline{u}_i, \underline{S}_i - Hp_i\rangle
    + \langle \overline{u}_i, Hp_i - \overline{S}_i\rangle
\end{align}
Since strong duality holds here, solving \hyperref[eq:D1_exp]{(D1)} is equivalent to solving
\begin{subequations}\label{eq:step1_dual}
\begin{align}
    &\max_{u_i \succeq 0} ~~\psi(u_i), ~~\text{where}\\
    \psi(u_i) \coloneqq &\min_{p_i}\left\{
    \mathcal{L}(p_i, u_i):
    p_i \in [\underline{P}_i, \overline{P}_i]
    \right\}
\end{align}
\end{subequations}
Define
$
p^\dagger(u_i) \coloneqq \arg\min_{p_i}\left\{
    \mathcal{L}(p_i, u_i):
    p_i \in [\underline{P}_i, \overline{P}_i]
    \right\}
$, which is a singleton, and its only element admits a closed-form\footnote{For simplicity, we do not distinguish the singleton set and its only element in notation.}
\begin{equation}\label{eq:p_dagger_D1}
    p^\dagger(u_i) = \text{clip}\left(
\frac{1}{\kappa}\left[-c_i + H^\top (\underline{u}_i - \overline{u}_i)\right],
\underline{P}_i, \overline{P}_i
\right)
\end{equation}

Note that $\psi(u_i)$ is essentially the dual function of \eqref{eq:D1_exp}, which is concave.
Therefore, we can solve \eqref{eq:step1_dual} by \emph{projected gradient ascent},
i.e., for step size $\eta > 0$, update $u_i$ to a new iterate $u_i^+$ according to:
\begin{subequations}\label{eq:step_1_update}
\begin{equation}
u_i^+ =  \text{Proj}_{\succeq 0}\left(u_i + \eta\nabla\psi(u_i)\right) = \left[u_i + \eta\nabla\psi(u_i)\right]^+
\end{equation}
where by Danskin's theorem \cite[Prop. B.25(a)]{bertsekas1999nonlinear},
\begin{equation}
    \nabla\psi(u_i) = \frac{\partial \mathcal{L}}{\partial u_i}\Big|_{p^\dagger} = \begin{bmatrix}
        \underline{S}_i - H p^\dagger(u_i)\\
        H p^\dagger(u_i) - \overline{S}_{i}
    \end{bmatrix}
\end{equation}
\end{subequations}

Denote the $k$-th iterate as $u_i^{\langle k \rangle}$, and $p_i^{\langle k \rangle} \coloneqq p^\dagger(u_i^{\langle k \rangle})$.
Suppose $\Phi(p_i)$ is the objective function of \hyperref[eq:D1_exp]{(D1)}, and $\Phi^\star$ is its optimal value.
\begin{theorem}[Convergence of Alg.\ref{alg:step1}]\label{thm:S1_conv}
Dual function $\psi(u_i)$ is $L_\psi$-smooth with $L_\psi=\frac{2}{\kappa}\norm{H}_2^2$, where
$\norm{H}_2^2 \approx 0.405 T^2$ is the (induced) $2$-norm of matrix $H$.

Therefore, under Alg.\ref{alg:step1} with step size $\eta \in (0, \frac{1}{L_\psi}]$,
\begin{subequations}
\begin{itemize}
    \item The dual objective values $\{\psi(u_i^{\langle k \rangle})\}_{k\ge 0}$ converge to $\Phi^\star$ in $\mathcal{O}(1/k)$, i.e.,
    \begin{equation}\label{eq:S1_conv_dual}
        \Phi^* - \psi(u_i^{\langle k \rangle}) \le \left(2\eta k\right)^{-1} \norm{u_i^{\langle 0 \rangle} - u_i^\star}^2
    \end{equation}
    where $u_i^\star$ is any dual optimal solution.
    \item The primal points $\{p_i^{\langle k \rangle}\}_{k\ge 0}$ converge to $p_i^\star$ in $\mathcal{O}(1/k)$, i.e.,
    \begin{equation}\label{eq:S1_conv_primal}
    \norm{p_i^{\langle k \rangle}) - p_i^\star}^2 \le \left(\widetilde{\eta} k\right)^{-1} \norm{u_i^{\langle 0 \rangle} - u_i^\star}^2
    \end{equation}
    where $p_i^\star$ is the unique primal optimal solution and $\widetilde{\eta} \coloneqq \eta \kappa$.
\end{itemize}
\end{subequations}

\end{theorem}
\begin{proof}
    See Appendix~\ref{appx:s1_solver}.
\end{proof}

\begin{remark}
    Note that the Lipschitz constant $L_\psi$ tends to be huge (e.g., $L_\psi \approx 2 \times 10^7$ when $T = 168$ and $\kappa = 10^{-3}$, the value used in Sec.~\ref{sec:solver_analysis}).
    Hence, for convergence, it requires a tiny step size $\eta$ if we keep it constant, and as a result, the convergence rate tends to be slow (note $\eta$ on both RHS of \eqref{eq:S1_conv_dual} and \eqref{eq:S1_conv_primal}).
    We explore \textsf{Adam} \cite{kingma2014adam}, which is an adaptive gradient-based optimizer combining momentum with parameter-wise step size scaling
    that has been widely used to optimize deep neural networks.
    We find it can substantially accelerate convergence over vanilla gradient descent (ascent), so we incorporate it in our custom optimizer.
    Implementation details and empirical comparisons are included in Appendix~\ref{appx:s1_solver}.
\end{remark}

\textbf{Duality gap}\qquad
For any feasible primal variable $\widetilde{p}_i$ and any feasible dual variable $\widetilde{u}_i$ (i.e., $\widetilde{p}_i \in \X_i$, $\widetilde{u}_i \succeq 0$), by duality theory,
$\psi(\widetilde{u}_i) \le \Phi^\star \le \Phi(\widetilde{p}_i)$,
which provides a practical way to certify the (relative) optimality gap by the (relative) duality gap:
\begin{equation}\label{eq:S1_rel_duality_gap}
    \Delta(\widetilde{p}_i, \widetilde{u}_i) \coloneqq
    \frac{\Phi(\widetilde{p}_i) - \psi(\widetilde{u}_i)}{\max\{|\Phi(\widetilde{p}_i)|, |\psi(\widetilde{u}_i)|, 1\}}
\end{equation}
In particular, since $\psi(\widetilde{u}_i) \le \Phi^\star$ holds for \emph{any} $\widetilde{u}_i \succeq 0$, evaluating the (D1) terms of $Q(\lambda)$ by the dual values $\psi$, rather than by the objectives of inexact primal iterates (which over-estimate the minima), keeps $Q(\lambda)$ in \eqref{eq:weak_duality} a valid global lower bound at any inner tolerance.

\begin{algorithm}
\caption{Custom solver for the \hyperref[eq:D1_exp]{(D1)} best response}\label{alg:step1}
\begin{algorithmic}[1]\small
\Procedure{EV\_Indiv\_Update~}{$\eta, \epsilon, K$}
\State \textbf{Initialize} $u_i \gets 0, k \gets 0$
\Comment{optional: warm-start with $u_i^0$}
\State $p_i \gets  p^\dagger(u_i)$ by \eqref{eq:p_dagger_D1}
\While{True}
\State $g \gets \nabla\psi(u_i)$ by plugging in $p_i$, following \eqref{eq:step_1_update}
\vspace{5pt}
\State ${u}_i \gets \left[u_i + \eta g\right]^+$, %
$p_i \gets  p^\dagger(u_i)$ by \eqref{eq:p_dagger_D1}
\vspace{5pt}
\State $\Delta \gets \Delta(\texttt{Feas}(p_i), u_i)$ by \eqref{eq:S1_rel_duality_gap}
\If{$\Delta \le \epsilon$ or $k > K$}
\State \textbf{break}
\EndIf
\State $k \gets k+1$
\EndWhile
\State \Return $\texttt{Feas}(p_i), \Delta$
\EndProcedure
\end{algorithmic}
\end{algorithm}

\begin{remark}[Subproblems of the ADMM baseline]\label{rmk:admm_subproblems}
The per-EV update \hyperref[eq:step_1]{(S1)} of the ADMM baseline shares the structure of \eqref{eq:D1_exp}, namely a strongly convex quadratic over the same constraints, with modulus $1 + \kappa/\rho$ in place of $\kappa$.
Alg.~\ref{alg:step1} and Theorem~\ref{thm:S1_conv} apply verbatim after a one-line change of the closed form $p^\dagger$, with the reduction and the resulting constants given in Appendix~\ref{appx:s1_solver}.
The grid-side update \hyperref[eq:step_2]{(S2)} is likewise solved exactly, per feeder in $\mathcal{O}(T \ln T)$ by a sort-and-search procedure (Theorem~\ref{thm:S2} in Appendix~\ref{appx:S2_solver}), which makes (S2) negligible relative to (S1) in the overall ADMM runtime.
\end{remark}

\section{Software Implementations and Analysis}\label{sec:solver_analysis}
\subsection{Software Implementation}
We implement IPR and the ADMM baseline, with their custom subproblem solvers, in \texttt{Python (3.11.13)}. Task parallelization is handled by \texttt{Dask (2025.9.1)}, a Python library for flexible and powerful parallel and distributed computing. Batched subproblem solvers are implemented with \texttt{NumPy (2.3.3)}, which provides efficient matrix arithmetic. Full experiments were run on \texttt{Savio3} partition on the Berkeley's high performance computing cluster, which consists of dual-socket Intel Xeon Skylake 6130-based compute nodes (32 physical cores per node) with 96 GB DDR4 memory. We typically request 10 nodes for a run.

\subsection{IPR Converges to a Deployable Optimum in Fewer Iterations}

\begin{figure}[htbp]
    \centering
    \includegraphics[width=\columnwidth]{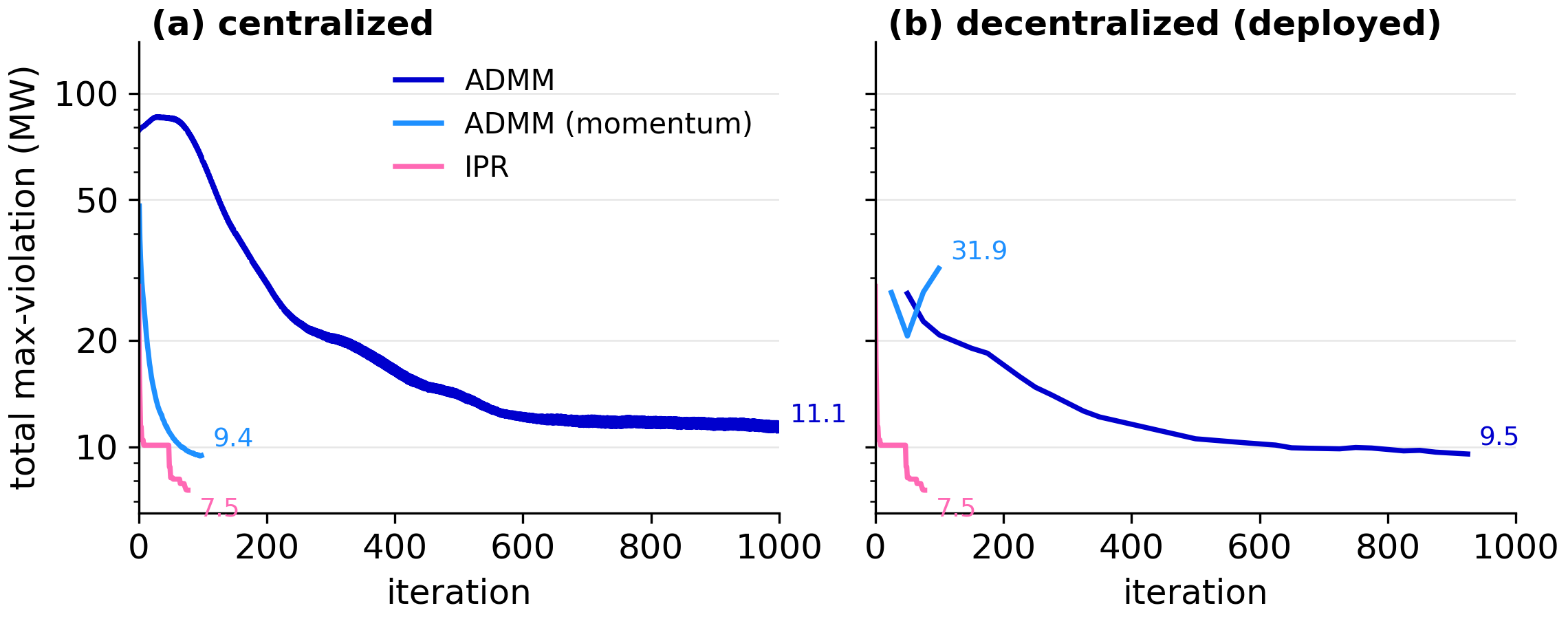}
    \caption{\textbf{Convergence of IPR and the ADMM baselines}
    \scriptsize\quad
    Total max-violation of \textbf{a.} the centralized allocation and \textbf{b.} the decentralized deployment, i.e., the aggregate best response to the price iterate, for \textsf{IPR} (Alg.~\ref{alg:ipr}), plain ADMM ($\rho=0.1$), and ADMM with the anchored dual momentum. IPR keeps no separate centralized iterate. Its allocation is the deployed response itself, so the same curve appears in both panels. The curve shows the best-so-far deployed response, which is the anytime iterate the method reports.
    }
    \label{fig:loss_curve}
\end{figure}

Figure~\ref{fig:loss_curve} compares the convergence of IPR with the two ADMM baselines on the full case study ($1.84$ million EVs, see Sec.~\ref{sec:case-study} for details). An iteration of either method sweeps the fleet once, one \hyperref[eq:D1]{(D1)} best response for IPR and one \hyperref[eq:step_1]{(S1)} solve for ADMM, and costs roughly one minute end-to-end on our hardware. For ADMM, the \emph{centralized} curve is computed from the primal iterate $p^{(k)}$, whereas the \emph{decentralized} curve is the violation of the aggregate best response to the price iterate $\lambda^{(k)}$, i.e., the number a price-based deployment would realize. For IPR the two coincide by construction.

Plain ADMM needs on the order of $1{,}000$ iterations ($\approx 20$~h) to reach a centralized violation of $11.1$~MW and a deployed violation of $9.5$~MW, with a relative duality gap \eqref{eq:SO_rel_gap} of about $12\%$. The anchored momentum variant accelerates the centralized allocation to $9.4$~MW within $100$ iterations, but its in-loop price deploys poorly ($32$~MW at iteration $100$), because the allocation and the price converge at different rates and no value-based ADMM signal controls the deployed number. IPR ascends the deployable price directly and dominates on the metric that matters. Starting from the shortage seed \eqref{eq:lse_price}, it reaches a deployed violation of $7.5$~MW, the best of the three, within $78$ iterations and with a certified relative optimality gap below $2\%$.

Appendix~\ref{appx:ipr_ablation} ablates the IPR price update against two Euclidean variants on the same seed, oracle, and line search, comparing both the deployed violation and the certified duality gap.
The ablation shows that the entropic geometry is what makes the ascent certify the solution it finds.

\begin{figure}[h!]
    \centering
    \includegraphics[width=\columnwidth]{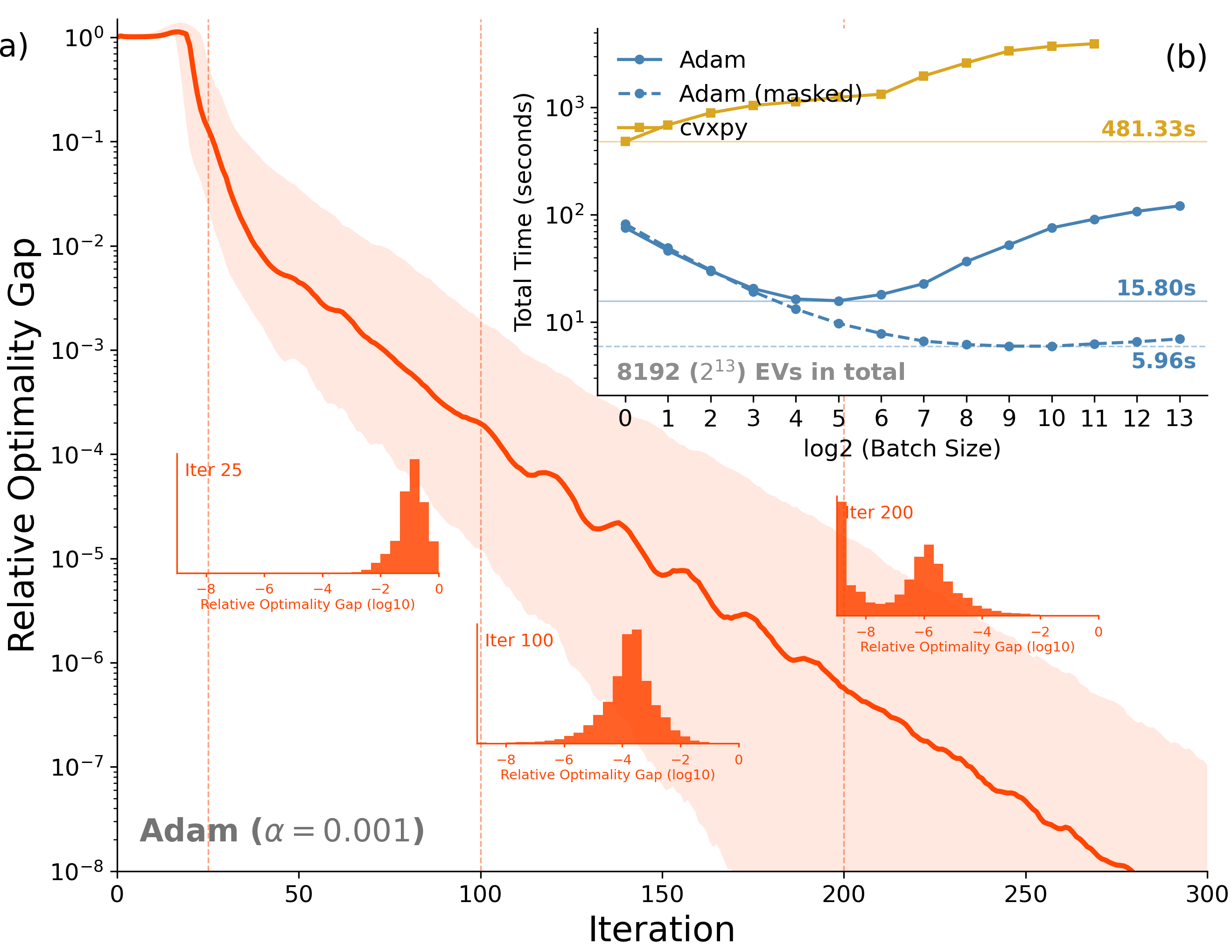}
    \caption{\textbf{Performance of custom (D1) solver}
    \scriptsize\quad
    \textbf{a.} \textit{main plot}: medians (solid line) and [10\%, 90\%] percentiles (shaded) of evaluated relative optimality gaps over Adam iterations.
    \textit{annotated histograms}:
    snapshots of relative duality gap distributions at 25$^\text{th}$, 100$^\text{th}$, and 200$^\text{th}$ iterations, respectively.
    \quad
    \textbf{b.}
    total solution time of 8192 instances vs varying batch sizes on one processing core.
    }
    \label{fig:solver_convergence}
\end{figure}

\subsection{Custom Subproblem Solvers Reduce Per-Iteration Solve Time}

We further evaluate the performance of our custom solver for subproblem \hyperref[eq:D1]{(D1)}. Since every IPR iteration invokes the best-response oracle once per EV, this oracle constitutes the primary computational bottleneck of the pipeline.
We randomly sample $8192$ EVs whose weekly energy demand exceeds $60$~kWh, and construct (D1) instances with $\kappa = 10^{-3}$. The price $\lambda$ is generated by the shortage-based rule \eqref{eq:lse_price} with $\beta=4$, applied to the aggregate load of a random feasible fleet profile scaled so that $30\%$ of the visited feeders are congested. EV $i$ then faces the price exposure $c_i = A_i^\top \lambda$ along its own itinerary.

\textbf{Accuracy}\quad
Figure~\ref{fig:solver_convergence}\textbf{a} reports the distribution of relative duality gaps (which are upper bounds of relative optimality gaps) across the sampled subproblem instances as a function of Adam iterations, with Adam configured as $(\eta,\beta_1,\beta_2)=(10^{-3},0.9,0.999)$. By iteration $25$, $3.8\%$ of instances achieve a relative duality (optimality) gap below $1\%$. By iteration $100$ and $200$, the proportions below $1\%$ gap increase to $97.8\%$ and $99.6\%$, while the proportions below $0.1\%$ gap reach $84.0\%$ and $98.7\%$, respectively.

\textbf{Speed}\quad
Figure~\ref{fig:solver_convergence}\textbf{b} compares the wall-clock time required to solve $8192$ instances on a single core using \texttt{CVXPY+OSQP} versus our Adam-based custom solver under different batch sizes $B$.

Already in \emph{serial} mode ($B=2^0$), the custom solver requires $75.8$~s against $481.3$~s for \texttt{CVXPY+OSQP}. As $B$ increases, \texttt{CVXPY}\footnote{For \texttt{CVXPY}, we sum the objectives and stack the constraints of a batch of instances to form a new problem.} slows down monotonically, and beyond $B=2^{11}$ ($3935$~s) a single batch exceeds a $1000$-s limit.
Our custom solver instead benefits from batching up to $B=2^5$ ($15.8$~s), before the computation becomes memory-bound at larger $B$.

Moreover, since the per-instance optimality gap is certified along the iterations, a \emph{masking} strategy excludes already-converged instances from subsequent iterations. This flattens the dependence on $B$ (the total time stays within $6$--$7$~s for all $B \ge 2^7$) and lowers the best total runtime to $5.96$~s, a $98.8\%$ reduction relative to \texttt{CVXPY+OSQP}.

\section{Real-world Case Studies}\label{sec:case-study}
\subsection{Data and Setups}

We conduct a case study in the San Francisco Bay Area, a national leader in EV adoption. Our primary data sources include grid hosting-capacity data from PG\&E \cite{pge_ica_map} and synthetic population and mobility data from Replica \cite{replica_population_2024q4,replica_trip_table_2024q4}.
Fig.~\ref{fig:data} provides an overview of the two primary data sources.
We evaluate the total max-violation across $1303$ distribution feeders induced by charging $1.84$ million EVs to satisfy their mobility demand over one week ($T = 168$ hourly time slots, i.e., $\Delta t = 1$~h). We solve \hyperref[eq:opt_SO]{(SO)} via IPR, which unveils the full potential of MAC, and compare the results against a variety of flexibility and coordination baselines.
\begin{figure}[h!]
    \centering
    \includegraphics[width=1.0\columnwidth]{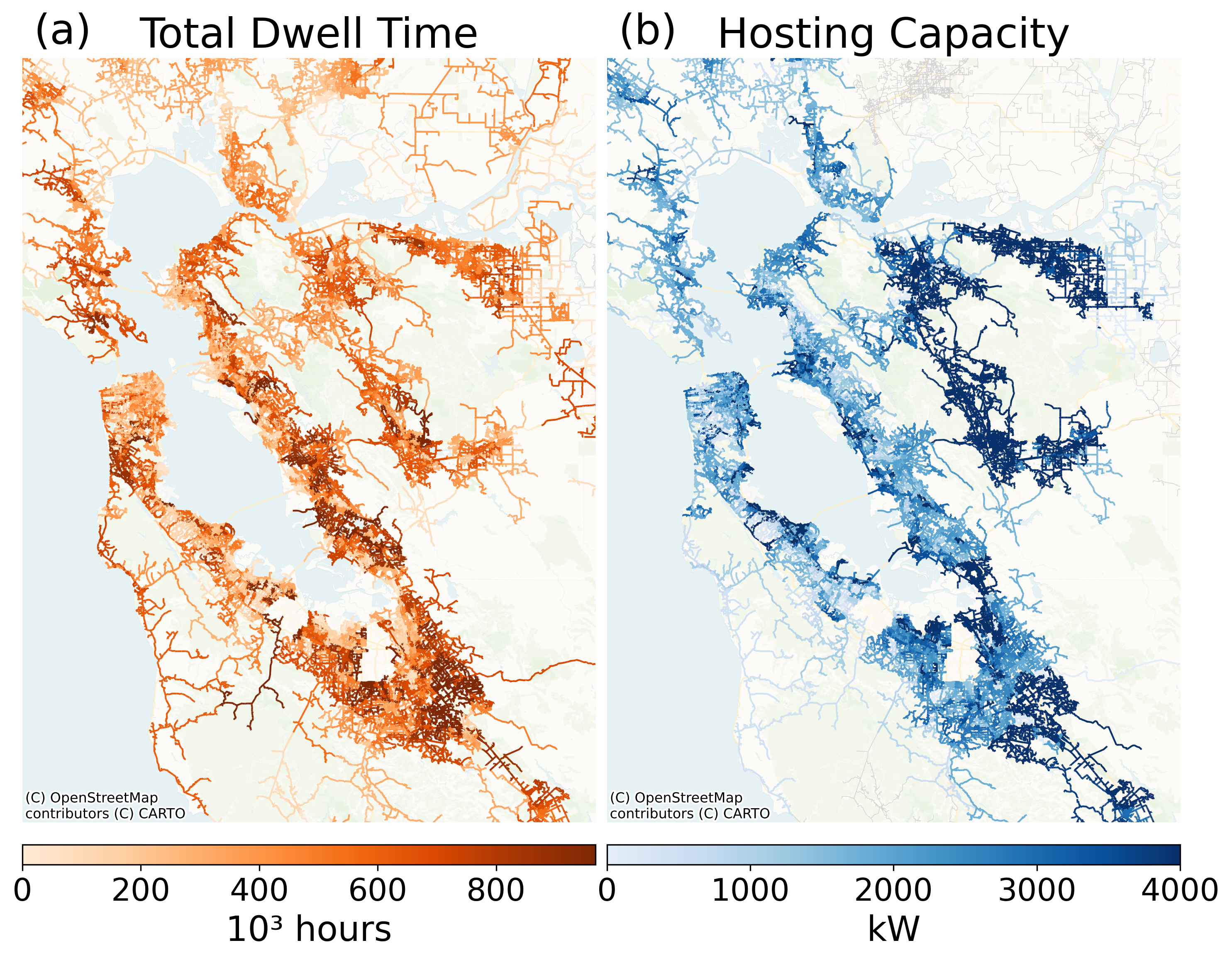}
    \caption{\textbf{Mobility and hosting capacity data}
    \scriptsize
    \quad We visualize aggregated metrics at the feeder level.
    \quad
    \textbf{a.} sum of dwell time (EV-hours) of all EVs over the one-week horizon at locations served by each distribution feeder.
    \quad
    \textbf{b.} time-averaged hosting capacity of distribution feeders.
    }
    \label{fig:data}
\end{figure}

\textbf{Grid hosting capacity}\quad
PG\&E publishes Integration Capacity Analysis (ICA), also referred to as \emph{hosting capacity}, for primary distribution line sections \cite{pge_ica_map}. ICA reports the maximum additional DER capacity that can be accommodated without violating operational limits, evaluated for a 24-hour profile on two representative days (high-load and low-load) under multiple violation criteria (e.g., voltage and thermal constraints). In our study, we use the 24-hour profile for a high-load day in December and consider both voltage and thermal constraints.
We adapt the method in \cite{navidi2023coordinating} to map line section-level hosting capacity to the feeder level. Specifically, each feeder's hosting capacity is computed as the population-weighted average of the hosting capacity of all line sections served by that feeder. Since the hosting capacity of a line section assumes load added at that section alone, and the capacity of a feeder is limited by its most constrained section, the residual violations reported below are likely on the optimistic side (cf.\ the limitations in Sec.~\ref{sec:conclusion}). We restrict attention to feeders within the boundaries of the nine Bay Area counties, resulting in a total of $1303$ feeders. For stays within the region, we assume access to a Level-2 charger at each stop, whose power can be continuously modulated between $0$ and its rated $10$~kW. For stays outside this region, we assume charging is unavailable.

\textbf{EV drivers and itineraries}\quad
Replica provides nationwide synthetic population \cite{replica_population_2024q4} and mobility \cite{replica_trip_table_2024q4} data by fusing multiple raw data sources using state-of-the-art behavioral and statistical modeling methods. We select individuals who either live or work in the Bay Area and sample EV drivers up to $30\%$ of the total driver population following a demographic-transition procedure that captures the gradual shift in adopter composition (see Appendix~\ref{appx:adopter_sampling}). Replica provides individual itineraries for two representative days (Thursday and Saturday). We extend these to continuous week-long itineraries using a cluster-based bootstrapping approach \cite{wang2026assessing}, summarized in Appendix~\ref{appx:traj_synth}.
After selection and preprocessing, our dataset contains $1.84 \times 10^6$ EV drivers with a total of $37.34 \times 10^6$ trips. Each EV is modeled with an $80$~kWh usable battery ($\approx\!267$-mile range at $0.3$~kWh/mile, representative of the best-selling models). The per-slot driving energy $D_{i,t}$ is derived from the traveled distance at the same efficiency.
Each stay location is matched to its nearest distribution line section, which determines the feeder onto which the charging demand is assigned.

\subsection{Baselines and Metrics}\label{sec:baselines}
To benchmark \textsf{MAC}, we compare it against a spectrum of baselines organized along two orthogonal axes. The first axis is how charging \emph{flexibility} is modeled. \emph{Session-based} policies fix the energy delivered in each charging session by a local rule, whereas \emph{mobility-aware} policies treat energy as fungible across the entire weekly itinerary, subject only to mobility feasibility. The second axis is the \emph{coordination} used, ranging from a fixed rule or a one-shot response to a price signal to a converged optimization of the system objective. This yields four rule-based policies and the two coordinated counterparts, summarized in Table~\ref{tab:baselines}. Formal definitions and the evaluation procedure are given in Appendix~\ref{appx:baselines}. The session-based rules are reported in both a front-loaded and a uniform (\textsf{u}-prefixed) intra-stay dispatch.

\textbf{\textsf{ASAP+}}\quad ``as-soon-as-possible'' rule: whenever an EV arrives at a stop with a state of charge below $50\%$, it charges at the maximum available power and stops once its battery is full or already holds just enough energy to complete all remaining trips and restore its terminal state of charge. Arriving at or above $50\%$, it charges only what the remaining itinerary strictly requires, which is nothing at most stops. It is the uncontrolled-immediate anchor of the spectrum and adapts the behavioral model of \cite{li2024impact}.

\textbf{\textsf{ASAN}}\quad ``as-soon-as-needed'' rule: at each stop it charges the least of a full battery, the energy needed to reach the next long stay (duration $\ge 3$~h) holding a $50\%$ reserve, and the energy its terminal target still requires. It captures a range-anxious driver who tops up for the upcoming leg rather than to full.

\textbf{\textsf{FLEX}}\quad pins each session's energy to a rule-based (\textsf{ASAP+} or \textsf{ASAN}) allocation and optimizes only the intra-stay \emph{timing} of that energy against the feeder objective. It is the globally optimized baseline for session-based flexibility.

\textbf{\textsf{minPeak}}\quad each EV \emph{individually} minimizes the peak of its own charging profile $\max_t p_{i,t}$ subject to mobility feasibility. Since every EV spreads its own demand as flat as its mobility allows, the policy acts as a greedy heuristic for shaving the aggregate peak while requiring no coordination among EVs.

\textbf{\textsf{PR}}\quad each EV best-responds \emph{once} to an endogenous congestion price inferred from the \textsf{ASAP+} overload. Unlike a feeder-invariant time-of-use price that only relocates the peak, \textsf{PR} prices where and when the grid is actually short.

\textbf{Evaluation metrics}\quad For each policy we aggregate the resulting feeder loads $l_{s,t}$ and report three violation metrics:
\begin{subequations}\label{eq:metrics}
\begin{align}
\mathrm{TV}_{\max}  &= \sum\nolimits_s \max\nolimits_t\, [l_{s,t}-C_{s,t}]^+, \label{eq:TVmax}\\
\mathrm{TV}_{\mathrm{avg}} &= \sum\nolimits_s \tfrac{1}{T}\sum\nolimits_t\, [l_{s,t}-C_{s,t}]^+, \label{eq:TVavg}\\
\#\mathrm{TV} &= \big|\{ s:\ \max\nolimits_t [l_{s,t}-C_{s,t}]^+ > 0.3\,\overline{C}_s \}\big|, \label{eq:nTV}
\end{align}
\end{subequations}
with $\overline{C}_s=\frac{1}{T}\sum_t C_{s,t}$. The \emph{total max-violation} $\mathrm{TV}_{\max}$ (in MW) is the objective \eqref{eq:Jgrid} that \textsf{MAC} directly minimizes. The \emph{total average-violation} $\mathrm{TV}_{\mathrm{avg}}$ (in MW) weights sustained rather than peak overload, and is more representative of cumulative thermal aging. The \emph{overloaded-feeder count} $\#\mathrm{TV}$ counts feeders whose peak overload is substantial relative to their own capacity. These metrics are planning-level quantities. Each sums over feeders an exceedance of the hourly-averaged feeder load, so they measure sustained overload over the week rather than instantaneous power at any single location or time. All metrics are evaluated at the hourly resolution of the model. Appendix~\ref{appx:subhourly} quantifies the sub-hourly effect by re-placing each policy's solution on grids down to $5$~min using the minute-resolved stay windows, and shows that the ordering in Table~\ref{tab:baselines} is unchanged at every resolution.

\providecommand{\tbd}{\textcolor{red}{tbd}}
\begin{table}[htbp]
\centering
\caption{\textbf{Baseline comparison across the coordination spectrum}}
\label{tab:baselines}
\resizebox{\columnwidth}{!}{%
\begin{tabular}{@{}lrrr@{\hskip 26pt}lrrr@{}}
\toprule
\multicolumn{4}{c}{\textbf{Rule-based / heuristic}} & \multicolumn{4}{c}{\textbf{Coordinated optimization}}\\
\cmidrule(r){1-4}\cmidrule(l){5-8}
 & $\mathrm{TV}_{\max}$ & $\mathrm{TV}_{\mathrm{avg}}$ & $\#\mathrm{TV}$ & & $\mathrm{TV}_{\max}$ & $\mathrm{TV}_{\mathrm{avg}}$ & $\#\mathrm{TV}$\\
\midrule
\multicolumn{8}{@{}l}{\textit{Session-based}}\\
\textsf{ASAP+}  & 12,439.6 & 246.8 & 1052 & \multirow{2}{*}{\textsf{FLEX}} & \multirow{2}{*}{1,792.7} & \multirow{2}{*}{136.3} & \multirow{2}{*}{623}\\
\textsf{uASAP+} & 3,590.4 & 168.7 & 791 &                                &                        &                        & \\
\textsf{ASAN}   & 12,441.3 & 185.0 & 1053 & \multirow{2}{*}{\textsf{FLEX}} & \multirow{2}{*}{744.3} & \multirow{2}{*}{76.8} & \multirow{2}{*}{423}\\
\textsf{uASAN}  & 1,910.2 & 101.6 & 656 &                                &                        &                        & \\
\midrule
\multicolumn{8}{@{}l}{\textit{Mobility-aware}}\\
\textsf{minPeak} & 79.4 & 31.2 & 127 & \multirow{2}{*}{\textsf{MAC}} & \multirow{2}{*}{8.1} & \multirow{2}{*}{4.5} & \multirow{2}{*}{49}\\
\textsf{PR}      & 102.6 & 11.7 & 132 &                               &                        &                        & \\
\bottomrule
\end{tabular}%
}
\end{table}

\subsection{Findings and Implications}
Overall, \textsf{MAC} reduces feeder overloading across the Bay Area by orders of magnitude relative to the unmanaged \textsf{ASAP+} charging mode. Under \textsf{MAC}, the realized system-wide total max-violation is $8.1$ MW (the IPR deployed response, certified within $2.9\%$ of the optimum), whereas \textsf{ASAP+} results in $12,440$ MW. This gap indicates that leveraging mobility-aware flexibility through full charging coordination can nearly eliminate overload-driven upgrade needs at the regional scale, which would otherwise translate to roughly $5.0$ billion USD in upgrade costs, based on the median feeder upgrade cost estimates (USD/kW) disclosed in PG\&E DDOR \cite{elmallah2022can}. In short, coordination converts what appears to be a system-wide capacity crisis under unmanaged charging into a largely manageable operational problem.

\begin{figure}[h!]
    \centering
    \includegraphics[width=\columnwidth]{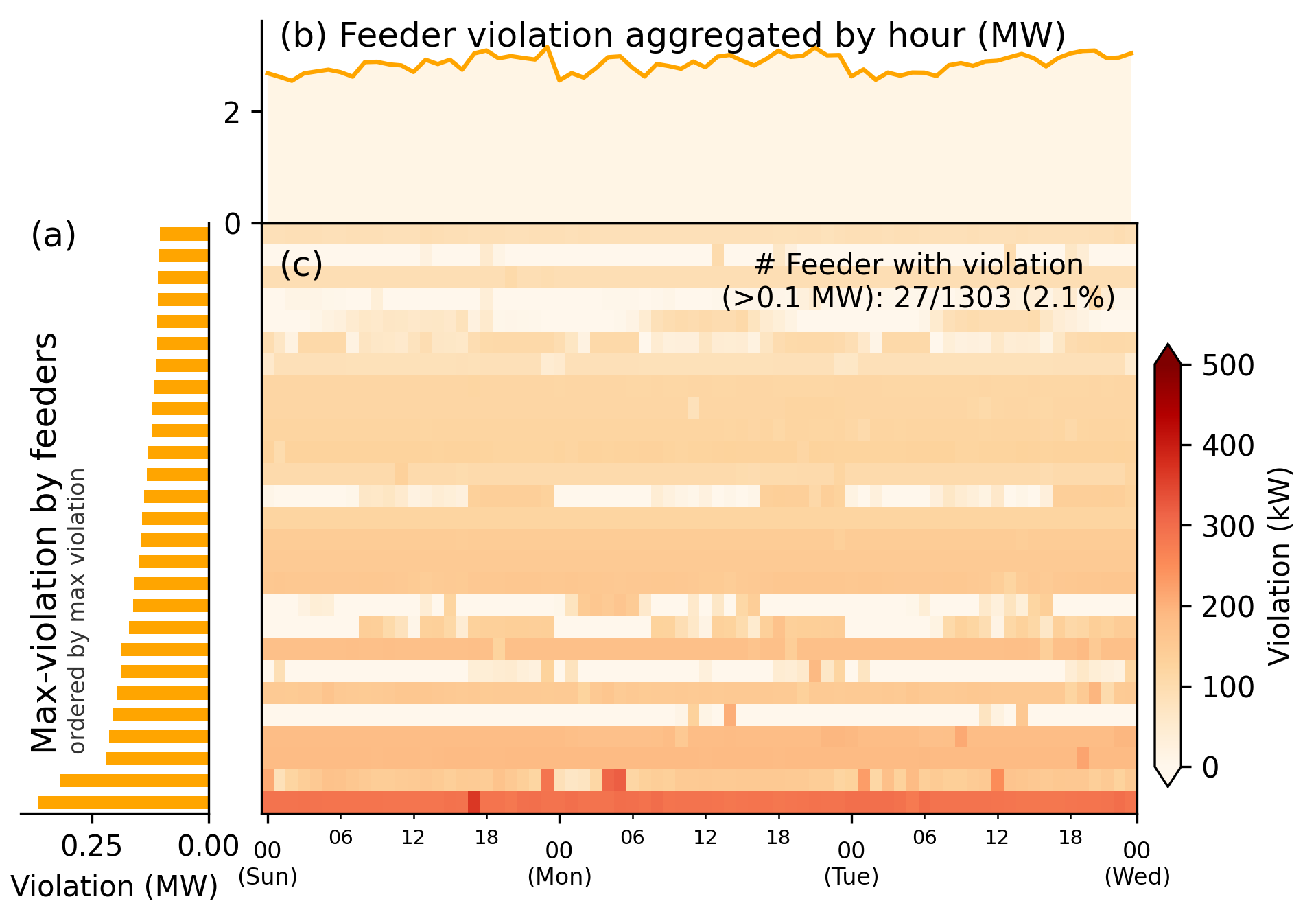}
    \caption{\textbf{Violation analysis under charging coordination}
    \scriptsize
    \quad \textbf{a.}
    max-violation (i.e., upgrade need) of individual feeders.
    \quad \textbf{b.}
    total overloaded demand in three consecutive days.
    \quad \textbf{c.}
    Heatmap of feeder-level overloading. Only feeders with $>0.1$ MW violation are listed.
    }
    \label{fig:line_violation}
\end{figure}

\begin{figure*}[htbp]
    \centering
    \includegraphics[width=0.94\textwidth]{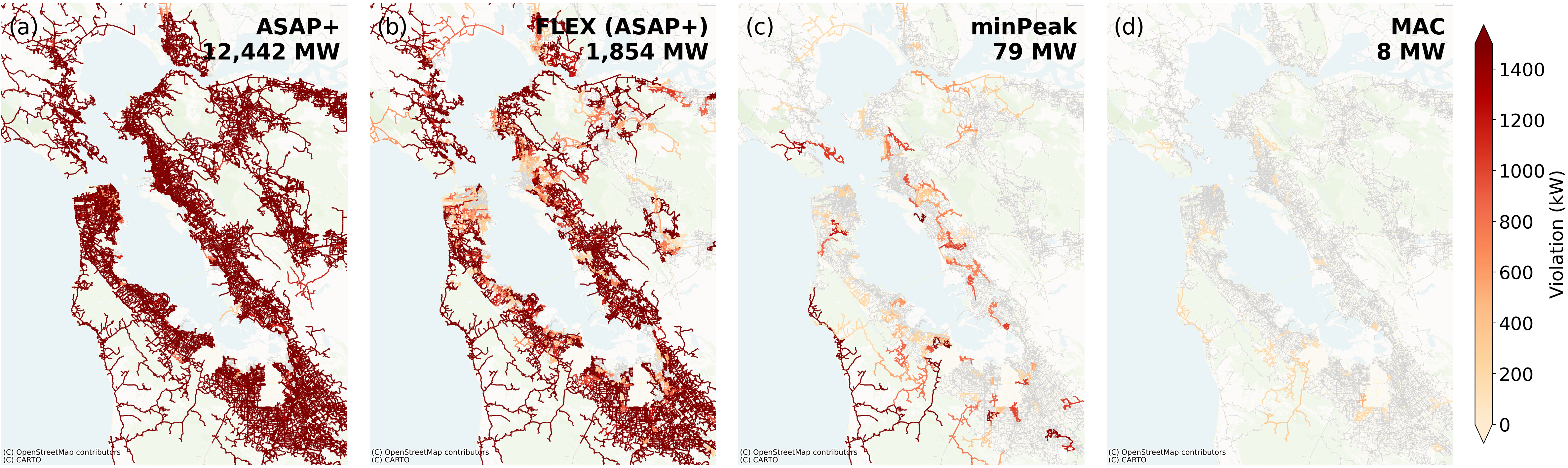}
    \caption{\textbf{Geospatial mapping of feeder violation}
    \scriptsize
    \quad
    Colors on feeder lines indicate the magnitude of realized max-violation under \textbf{a.} ASAP+, \textbf{b.} FLEX (ASAP+), \textbf{c.} minPeak, and \textbf{d.} MAC, following the flexibility and coordination spectrum of Fig.~\ref{fig:G-H}\textbf{c}. Each panel annotates the policy's total max-violation.
    }
    \label{fig:violation_map}
\end{figure*}

The reduction is not only system-wide but also broadly distributed across feeders, as shown in Figure~\ref{fig:line_violation}. Using a max-violation threshold of $0.1$ MW, $83\%$ of feeders ($1076/1303$) would require upgrades under \textsf{ASAP+}, whereas this fraction drops to only $2.1\%$ ($27/1303$) under \textsf{MAC}. By the substantial-overload count $\#\mathrm{TV}$ of Table~\ref{tab:baselines}, the drop is from $1052$ to $49$ feeders. Figure~\ref{fig:violation_map} visualizes the overloaded feeders in each scenario, where color intensity reflects the magnitude of the inferred capacity upgrade requirement. Importantly, the remaining violations under \textsf{MAC} are sparse and localized, suggesting that the residual upgrade needs would likely be targeted rather than region-wide.

The two axes of Table~\ref{tab:baselines} decompose this reduction into the benefit of flexibility and that of coordination, traced by the four panels of Fig.~\ref{fig:violation_map}.
Along the coordination axis, fully optimizing the timing of session-based charging (\textsf{FLEX}) cuts the \textsf{ASAP+} violation from $12{,}440$~MW to $1{,}793$~MW, a $6.9\times$ reduction that still leaves a region-wide upgrade burden.
Along the flexibility axis, \textsf{minPeak} lets each EV spread its own demand across its weekly itinerary and reaches $79$~MW \emph{without any coordination among EVs}, outperforming the fully coordinated session-based \textsf{FLEX} by more than an order of magnitude.
Mobility-aware flexibility is therefore the larger lever, as it removes most of the overload that even perfectly timed session-based charging cannot.
Coordination completes the remaining step. By pricing where and when feeders are actually short, \textsf{MAC} brings the violation down to $8.1$~MW, another tenfold reduction.
The gap between \textsf{FLEX} and \textsf{MAC} isolates the contribution of mobility-aware itinerary coupling, and the gap between \textsf{minPeak} and \textsf{MAC} isolates the value of coordination given full flexibility.

\section{Conclusions and Future Works}\label{sec:conclusion}

The paper proposes a novel charging coordination paradigm, termed ``mobility-aware'' coordination, that couples charging decisions across multiple sessions over a week-long mobility horizon, unlocking substantially greater demand flexibility than conventional session-based approaches. It also derives, implements, and validates numerical optimization schemes, converting a coupled problem to one that can leverage the power of modern large-scale distributed computing systems. Key innovations include a certified price-response iteration with convergence guarantees, and custom subproblem solvers that significantly enhance the efficiency over standard distributed optimization pipelines.

The work characterizes the upper bound of grid benefits leveraging full coordination, which answers a critical question and provides insights to a range of stakeholders. Full-scale experiments of 1.84 million EV drivers across over 1300 distribution feeders demonstrate the great potential of eliminating overload-driven upgrade needs at the regional scale by the coordinated scheme.
For distribution utilities, the results offer a certified estimate of the upgrade investment that coordination can defer, and a locational price signal that can inform the design of managed-charging programs for EV drivers or aggregators.
For aggregators, the price-response pipeline and the custom solvers are directly reusable to schedule the EVs under their management.
In principle, the coordination can be implemented by posting the converged feeder-level price, computed ahead of time from forecast itineraries, under which each EV's own best response reproduces the coordinated schedule with no further exchange (Thm.~\ref{thm:SO_CE-2}).
Realizing it in practice raises questions beyond the present scope, including communication with chargers, the level of driver participation, and the robustness of the price to forecast errors and imperfect communication.

The framework and results also set up a baseline that further studies can refine in three directions: (1) incorporate more operational constraints that help characterize the demand flexibility potential in more realistic settings, as the present model assumes a charging opportunity at every dwell location and a static hosting capacity per feeder, and relaxing either assumption, by treating chargers as a second congestible resource or by imposing network power-flow constraints that couple feeders, requires new methodology; (2) propose more effective algorithms (may be approximate) and quantify its optimality gap by comparing with our results; (3) model driver behavior beyond the cooperative, price-responsive drivers assumed here, since frictions such as plug-in effort at every stop, walking distance to public chargers, and uncertainty in departure time and energy may affect compliance even though the coordination asks no driver to alter a trip or forgo energy.
We also encourage researchers to delve into the market mechanism behind MAC, e.g., how the avoided upgrade cost is shared with participating drivers, and how incomplete and/or uncertain mobility information is handled.

\renewcommand*{\bibfont}{\footnotesize} %
\printbibliography[title=References]

\bigskip
{\small
\noindent
\textbf{Declaration of LLM usage}\qquad
The authors used LLMs (primarily GPT 5 series, and Claude Opus 4.6) for enhancing the readability of the manuscript. The authors reviewed and edited the content as needed and take full responsibility for the content.
}

\newrefsection
\newpage

\appendices
\section{Proofs for the Iterated Price Response Method}\label{appx:ipr}

\subsection{Proof of Lemma~\ref{lem:Q_smooth}}
We first record an identity that converts per-EV sums into feeder-level weighted norms.
\begin{remark}\label{remark:sqrtN}
For any $M \in \R^{S\times T}$,
\begin{align}
    \sum\nolimits_i \norm{A_i^\top M}_2^2
     & = \sum_{i,t}\, M_{s,t}^2 \Big|_{s=\Gamma_{i,t}}
     = \sum_{s,t}\,  \sum_{i\in \I_{s,t}}\,  M_{s,t}^2\notag\\
     & = \sum_{s,t}\,  N_{s,t} \, M_{s,t}^2
     = \norm{M \odot \sqrt{N}}_F^2
\end{align}
where $\odot$ denotes the element-wise product between two matrices of the same dimensions, and $\sqrt{N}\in \R^{S\times T}\coloneqq (\sqrt{N_{s,t}})_{s\in \Sset, t \in \Tset}$.
\end{remark}

\begin{proof}[Proof of Lemma~\ref{lem:Q_smooth}]
Concavity is immediate, since by \eqref{eq:SO_Lagrangian} $Q$ is an infimum of functions affine in $\lambda$.
Since $\kappa>0$, the (D1) objective is $\kappa$-strongly convex, so its minimizer $p_i(\lambda)$ over the nonempty closed convex set $\X_i$ is unique. Differentiability and the gradient expression \eqref{eq:gradQ} then follow from Danskin's theorem \cite[Prop.~B.25]{bertsekas1999nonlinear}.

For the Lipschitz constant, fix $\lambda, \lambda^\prime \in \Lambda$ and let $\Delta p_i \coloneqq p_i(\lambda) - p_i(\lambda^\prime)$, $\Delta\lambda \coloneqq \lambda - \lambda^\prime$.
The first-order optimality conditions of the two best responses read
\begin{align*}
&\langle \kappa\, p_i(\lambda) + A_i^\top \lambda,~ p_i(\lambda^\prime) - p_i(\lambda)\rangle \ge 0\\
&\langle \kappa\, p_i(\lambda^\prime) + A_i^\top \lambda^\prime,~ p_i(\lambda) - p_i(\lambda^\prime)\rangle \ge 0
\end{align*}
Adding the two inequalities yields
$$\kappa \norm{\Delta p_i}_2^2 \le \langle A_i^\top \Delta\lambda,\, -\Delta p_i\rangle \le \norm{A_i^\top \Delta\lambda}_2 \norm{\Delta p_i}_2,$$
hence
\begin{equation}\label{eq:BR_lipschitz}
    \norm{\Delta p_i}_2 \le \frac{1}{\kappa} \norm{A_i^\top \Delta\lambda}_2
\end{equation}
Since $\big(\nabla Q(\lambda) - \nabla Q(\lambda^\prime)\big)_{s,t} = \sum_{i \in \I_{s,t}} \Delta p_{i,t}$ by \eqref{eq:gradQ}, the Cauchy--Schwarz inequality gives
\begin{align}
&\norm{\nabla Q(\lambda) - \nabla Q(\lambda^\prime)}_F^2\notag\\
&= \sum_{s,t} \Big(\sum\nolimits_{i\in\I_{s,t}} \Delta p_{i,t}\Big)^2
\le \sum_{s,t} N_{s,t} \sum\nolimits_{i\in\I_{s,t}} \Delta p_{i,t}^2\notag\\
&\le N_{\max} \sum\nolimits_i \norm{\Delta p_i}_2^2
\le \frac{N_{\max}}{\kappa^2} \sum\nolimits_i \norm{A_i^\top \Delta\lambda}_2^2\notag\\
&= \frac{N_{\max}}{\kappa^2} \norm{\Delta\lambda \odot \sqrt{N}}_F^2
\le \frac{N_{\max}^2}{\kappa^2} \norm{\Delta\lambda}_F^2
\end{align}
where the third inequality uses \eqref{eq:BR_lipschitz} and the following equality uses Remark~\ref{remark:sqrtN}.
Taking square roots completes the proof.
\end{proof}

\subsection{Proof of Theorem~\ref{thm:ipr_conv}}
\begin{proof}
\emph{Part (a).}\quad
By Lemma~\ref{lem:Q_smooth}, $-Q$ is convex with an $L_Q$-Lipschitz gradient on the closed convex set $\Lambda$, so \eqref{eq:ipr_generic} with step $1/L_Q$ is the classical projected gradient method, for which the value bound \eqref{eq:ipr_conv_dual} is standard (see, e.g., \cite[Ch.~10]{beck2017first}).
The same $\mathcal{O}(1/k)$ rate is retained by ascent variants whose accepted iterates satisfy the sufficient-increase condition enforced by backtracking. For the entropic mirror ascent of Sec.~\ref{sec:ipr_entropic}, the Euclidean quantity $\norm{\lambda^{(0)}-\lambda^\star}_F^2$ is replaced by the corresponding Bregman (Kullback--Leibler) diameter of $\Lambda$, which is $\mathcal{O}(S \log T)$ (see \cite[Ch.~9]{beck2017first}).

\emph{Part (b).}\quad
Fix $\lambda \in \Lambda$ and split the Lagrangian \eqref{eq:SO_Lagrangian} as
\begin{equation*}
\mathcal{L}(p, l, \lambda) = \sum\nolimits_i \left[\Jev_i(p_i) + \langle A_i^\top\lambda, p_i\rangle\right] + \left[\Jgrid(l) - \langle\lambda, l\rangle\right]
\end{equation*}
For each $i$, the bracketed EV term is $\kappa$-strongly convex in $p_i$ and minimized over $\X_i$ at the best response $p_i(\lambda)$, with optimal value $\xi_i(\lambda)$. Strong convexity gives, for any $p_i \in \X_i$,
\begin{equation*}
\Jev_i(p_i) + \langle A_i^\top\lambda, p_i\rangle \ge \xi_i(\lambda) + \frac{\kappa}{2}\norm{p_i - p_i(\lambda)}_2^2
\end{equation*}
The bracketed grid term is bounded below by its minimum $-\langle\lambda, C\rangle$ (Lemma~\ref{lem:D2}).
Summing over $i$ and using \eqref{eq:Q_on_Lambda}, for any $p \in \Pi_i \X_i$ and any $l$:
\begin{equation}\label{eq:L_lower}
\mathcal{L}(p, l, \lambda) \ge Q(\lambda) + \frac{\kappa}{2}\sum\nolimits_i \norm{p_i - p_i(\lambda)}_2^2
\end{equation}
Now evaluate the left side at the primal optimum $(p^\star, l^\star)$ of \hyperref[eq:opt_SO]{(SO)}. Since $l^\star = \sum_i A_i p_i^\star$, the coupling term of \eqref{eq:SO_Lagrangian} vanishes, so
\begin{equation*}
\mathcal{L}(p^\star, l^\star, \lambda) = J^\star = Q(\lambda^\star)
\end{equation*}
where the last equality is strong duality, which holds by the refined Slater condition (cf.\ the Remark in Sec.~\ref{sec:mac}).
Substituting into \eqref{eq:L_lower} and rearranging yields \eqref{eq:ipr_conv_primal}. Combining with part (a) gives the $\mathcal{O}(1/\sqrt{k})$ rate for the deployed best response.
\end{proof}

\subsection{Ablation of the Price-Update Geometry}\label{appx:ipr_ablation}

Fig.~\ref{fig:ipr_ablation} compares three price updates that share IPR's shortage seed, best-response oracle, and backtracking line search, so that only the update geometry differs: the entropic update \eqref{eq:polish} (\textsf{IPR}); Euclidean projected gradient ascent \eqref{eq:ipr_generic} with the same per-feeder gradient normalization (\textsf{vanilla}); and the Euclidean update with a single global scale fixed at the first iteration (\textsf{vanilla, unscaled}).
They are compared on the two quantities the method reports, the deployed violation and the running certified relative duality gap \eqref{eq:SO_rel_gap}.

With the per-feeder normalization kept, Euclidean ascent transiently amplifies the deployed violation past $150$~MW, since additive mass concentrated on peak slots herds the responding fleet into neighboring slots, and its line search stalls after $51$ iterations.
The best deployed violation along the way, $7.9$~MW, is within $4\%$ of IPR, but the dual value stalls $23\%$ below the entropic optimum, leaving a certified gap of $28\%$. The multiplicative geometry is what makes the ascent \emph{certify} the solution it finds.
Removing the normalization as well is far more damaging.
One global step size cannot serve feeders whose shortages span orders of magnitude. The ascent oscillates for the full $200$-iteration budget and only reaches $11.2$~MW with a certified gap of $82\%$.

\begin{figure}[htbp]
    \centering
    \includegraphics[width=\columnwidth]{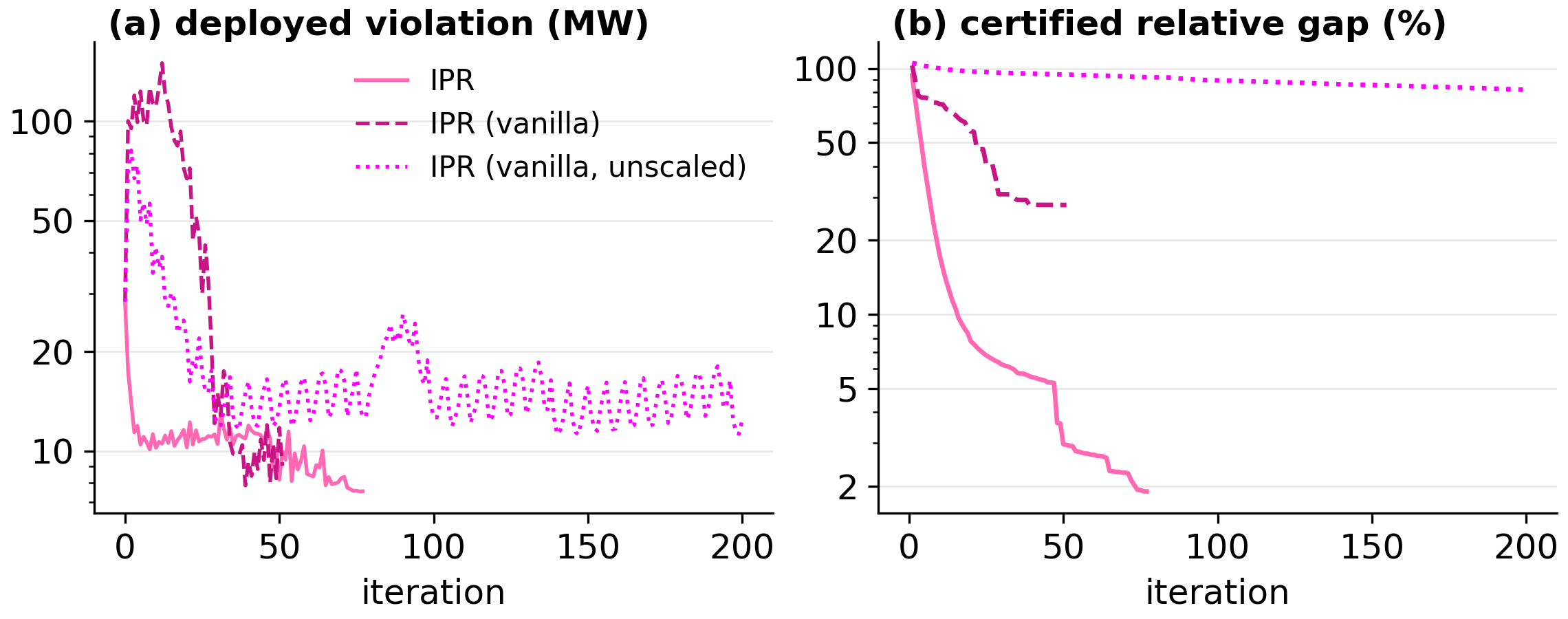}
    \caption{\textbf{IPR price-update ablation}
    \scriptsize\quad
    \textbf{a.} deployed total max-violation (raw iterates, with transients kept visible) and \textbf{b.} running certified relative duality gap \eqref{eq:SO_rel_gap}, for the entropic update and the two Euclidean variants. All three share the shortage seed, the (D1) best-response oracle, and the backtracking line search.
    }
    \label{fig:ipr_ablation}
\end{figure}

\section{Supplementary Notes on ADMM}\label{appx:admm}

\subsection{Derivation of the Reduced Updates \eqref{eq:admm_steps}}\label{appx:admm_deriv}
The (unscaled) augmented Lagrangian of \hyperref[eq:ADMM_0]{(P)} is:
\begin{align}\label{eq:Lag_unscale}
    \widetilde{L}_{\rho}(p, z, \lambda) \coloneqq
    & \sum\nolimits_{i}\, f_i(p_i) + g(\sum\nolimits_i z_i) \notag\\
    & \quad + \sum\nolimits_i \left(\langle\lambda_i, r_i\rangle_F + \frac{\rho}{2} \norm{r_i}_F^2\right)
\end{align}
whose scaled form is \eqref{eq:Lag_scale} in the main text.
The generic ADMM steps are: (S1) $p^{(k+1)} \gets \arg\min_{p}\, L_\rho(p, z^{(k)}, \mu^{(k)})$; (S2) $z^{(k+1)} \gets \arg\min_{z}\, L_\rho(p^{(k+1)}, z, \mu^{(k)})$; (S3) $\mu_i^{(k+1)} \gets \mu_i^{(k)} + r_i^{(k+1)}$, $\forall i$.
By plugging in the expression of the scaled Lagrangian \eqref{eq:Lag_scale} and omitting constants w.r.t. optimization variables, (S1) and (S2) solve the following concrete optimization problems, respectively:
\begin{align}
    \text{(S1)}~~& \min_{p_i}~  \frac{\rho}{2}\norm{A_i p_i - b_i^{(k)}}_F^2 + f_i(p_i), ~ \forall i\label{eq:step_1_full}\\
    \text{(S2)}~~& \min_{z}~ g\left(\sum\nolimits_i z_i\right) + \frac{\rho}{2} \sum\nolimits_i\norm{z_i - d_i^{(k)}}_F^2\label{eq:step_2_full}
\end{align}
where $b_i^{(k)}\coloneqq z^{(k)}_i - \mu_i^{(k)}$ and $d_i^{(k)} \coloneqq A_i p_i^{(k+1)} + \mu_i^{(k)}$.
It is obvious that we can optimize \eqref{eq:step_1_full} across different EV $i$ in parallel.
However, \eqref{eq:step_2_full} is currently not separable for $z_i$ over different $i$. Optimization of the problem as a whole is intractable for variable $z$ with dimension of $I \times ST$ (on the order of billions).

The optimization problem \hyperref[eq:ADMM_0]{(P)} can be classified into a generic class of problems known as \emph{sharing} \cite[Sec. 7.3]{boyd2011distributed}, for which highly efficient distributed updates can be derived.
The key idea is to reformulate \eqref{eq:step_2_full} as:
\begin{align}
    \min_{\overline{z}, z}~ & g(I \overline{z}) + \frac{\rho}{2} \sum\nolimits_i\norm{z_i - d_i^{(k)}}_F^2 \label{eq:step_2_zbar}\\
    \text{s.t.}~ & \sum\nolimits_i z_i = I \overline{z} \notag
\end{align}
and note that for any fixed $\overline{z} \in \R^{S \times T}$, the minimizer of
$$
    \min_{z}~ \frac{\rho}{2} \sum\nolimits_i\norm{z_i - d_i}_F^2\qquad
    \text{s.t.}~ \sum\nolimits_i z_i = I \overline{z}
$$
has an analytical solution $
    z^\star_i = \overline{z} + d_i - \overline{d}
$
where $\overline{d} = \frac{1}{I}\sum_i d_i$.
By plugging $z^\star_i$ into \eqref{eq:step_2_zbar}, \eqref{eq:step_2_full} has been reduced to solving for variable $\overline{z} \in \R^{S \times T}$:
\begin{equation}\label{eq:step_2_zbar_only}
    \overline{z}^{(k+1)} = \arg\min_{\overline{z}}~ g(I \overline{z}) + \frac{\rho I}{2} \norm{\overline{z} - \overline{d}^{(k)}}_F^2
\end{equation}

Consequently, by substituting $z_i^{(k+1)}$ into the (S3) multiplier update, we get:
\begin{align}
\mu_i^{(k+1)} & = \mu_i^{(k)} + \left(A_i p_i^{(k+1)} - z_i^{(k+1)}\right)\\
& = \left(\mu_i^{(k)} + A_i p_i^{(k+1)}\right) - z_i^{(k+1)}
= \overline{d}^{(k)} - \overline{z}^{(k+1)}\notag
\end{align}
which is a constant for all different $i$ (known as \emph{consensus}), i.e.,
\begin{equation}
\mu_i^{(k+1)} \equiv \overline{\mu}^{(k+1)} = \overline{\mu}^{(k)} + \frac{1}{I}\sum\nolimits_i r_i^{(k+1)}
\end{equation}
Hereafter, we simply denote $\overline{\mu}$ as $\mu$.
Finally, substituting $l = I\overline{z}$ and $D^{(k)} = I\overline{d}^{(k)} = \hat{l}^{(k+1)} + I\mu^{(k)}$ into \eqref{eq:step_2_zbar_only} yields the (S2) update of \eqref{eq:admm_steps}. The reduced (S1) and (S3) follow directly.
This completes the derivation.

\subsection{Proof of Theorem~\ref{thm:admm_conv}}
We verify that the problem \hyperref[eq:ADMM_0]{(P)} satisfies the two assumptions in \cite[Sec.3.2]{boyd2011distributed}:
\begin{itemize}
    \item[(A1)] Functions $f_i$'s and $g$ are  closed, proper, and convex.
    \item[(A2)] Strong duality holds for \hyperref[eq:ADMM_0]{(P)}, and dual optimal point is attained.
\end{itemize}

For (A1), note that $g(l) = \Jgrid(l)$ is a finite sum of pointwise maximum of convex functions, and clearly $\Jgrid(l) < +\infty$ for any $l \in \R^{S\times T}$, so function $g$ is closed, proper, and convex.
$f_i(p_i) = \Ind_{\X_i}(p_i) + \Jev_i(p_i)$, where $\Jev_i(p_i)$ is quadratic. Suppose $\X_i \neq \emptyset$, then $\Ind_{\X_i}(p_i)$ is proper, hence function $f_i$ is closed, proper, and convex.

For (A2), suppose $\X_i \neq \emptyset$ for all $i \in \Iset$, since $X_i$'s are closed convex sets, then $\text{ri}(\X_i) \neq \emptyset$ for all $i \in \Iset$, where $\text{ri}(\cdot)$ denotes the relative interior of a set,
and take any feasible $p \in \Pi_i \X_i$, let $l = \sum_i A_i p_i$, $l \in \text{ri}(\text{dom} g) = \R^{S\times T}$. Therefore, it satisfies the Slater's condition, so strong duality holds and dual optimal point is attained.
Note that these conditions hold for any $\kappa \ge 0$, including $\kappa = 0$.

\subsection{Anchored Dual Momentum and Warm Start}\label{appx:accel}
The accelerated ADMM pipeline evaluated in Sec.~\ref{sec:solver_analysis} (Fig.~\ref{fig:loss_curve}) augments the baseline \eqref{eq:admm_steps} with two dual-side modifications built from the shortage-to-price map \eqref{eq:lse_price}. Neither changes the objective or the subproblem solvers.
After each (S3) multiplier update, the scaled dual is blended toward the current iterate's shortage price,
\begin{equation}\label{eq:anchor}
    \mu^{(k+1)} \gets (1-a_k)\, \mu^{(k+1)} + a_k\, \widehat{\lambda}(\hat{l}^{(k+1)}) / \rho,
    \quad a_k = \frac{a}{1+k},
\end{equation}
with $a = 1/2$.
At $k=0$ this acts as a \emph{warm start}, since it initializes the dual at the natural magnitude of the congestion price, whereas the multiplier update \eqref{eq:step_3} climbs from $\mu = 0$ in $\mathcal{O}(1/I)$ steps and spends hundreds of iterations recovering that magnitude.
For $k \ge 1$ it acts as \emph{momentum}. The multiplier update integrates past residuals and adapts slowly to the moving shortage, while the anchor adds a proportional term that tracks it.
The Halpern-style annealing $a_k \to 0$ makes the perturbation vanish over iterations, so the anchored scheme does not alter the limit points of the unmodified ADMM.
Empirically, the anchor accelerates the centralized allocation by roughly an order of magnitude in iteration count (Fig.~\ref{fig:loss_curve}a), but it biases the price iterate itself (Fig.~\ref{fig:loss_curve}b). Obtaining a deployable price from the anchored run then requires a post-hoc IPR ascent (Sec.~\ref{sec:ipr_entropic}) from the ADMM price.
These modifications are retained as an ablation reference for the comparison in Sec.~\ref{sec:solver_analysis}.

\section{Proof of Theorem~\ref{thm:feas_oracle}}
\label{appx:feas_oracle}
\begin{proof}
We show both directions hold.

Let $S_{i,t} = (H p_i)_t = \sum_{\tau=1}^t p_{i,\tau}$.
\begin{enumerate}
\item $\X_i \neq \emptyset \Longrightarrow \underline{B}_i \preceq \overline{B}_i$:
Suppose $\underline{B}_i \npreceq \overline{B}_i$, then it is obvious that there exists no $p_i \in \R^T$ such that $\underline{B}_i \preceq Hp_i \preceq \overline{B}_i$.
Hence, for any $p_i \in \R^T$ there exists $t \in \Tset$ such that either $\underline{B}_{i,t} > S_{i,t}$, or $\overline{B}_{i,t} < S_{i,t}$.

Assume for some $t$, $\underline{B}_{i,t} > S_{i,t}$, and let $t^\prime =\arg\max_{\tau \ge t} \,\{\underline{S}_{i,\tau} - \overline{C}_{i,\tau}\}$. Then,
\begin{align}
\underline{B}_{i,t^\prime}
& = \overline{C}_{i, t^\prime} + \max_{\tau \ge t^\prime} \,\{\underline{S}_{i,\tau} - \overline{C}_{i,\tau}\}\\
& = (\overline{C}_{i, t} + \sum\nolimits_{\tau=t+1}^{t^\prime} \overline{P}_{i,\tau}) + \max_{\tau \ge t} \,\{\underline{S}_{i,\tau} - \overline{C}_{i,\tau}\}\notag\\
& = \underline{B}_{i,t} + \sum\nolimits_{\tau=t+1}^{t^\prime} \overline{P}_{i,\tau}\\
& > S_{i,t} + \sum\nolimits_{\tau=t+1}^{t^\prime} p_{i,\tau}
 = S_{i, t^\prime}\notag
\end{align}
On the other hand,
\begin{align}
\underline{B}_{i,t^\prime} & = \overline{C}_{i, t^\prime} + \max_{\tau \ge t^\prime} \,\{\underline{S}_{i,\tau} - \overline{C}_{i,\tau}\}\notag\\
& = \overline{C}_{i, t^\prime} + (\underline{S}_{i,t^\prime} - \overline{C}_{i,t^\prime}) = \underline{S}_{i,t^\prime}
\end{align}
Hence, $S_{i, t^\prime} < \underline{S}_{i,t^\prime}$, which indicates $p_i \notin \X_i$.

Similarly, if for some $t$, $\overline{B}_{i,t} < S_{i,t}$, then there exists $t^\prime \ge t$ such that $S_{i, t^\prime} > \overline{S}_{i,t^\prime}$, which indicates $p_i \notin \X_i$.

\item $\X_i \neq \emptyset \Longleftarrow \underline{B}_i \preceq \overline{B}_i$:
For any $p_i \in \R^T$,
if $\underline{P}_i \preceq p_i \preceq \overline{P}_i$ and $\underline{B}_i \preceq Hp_i \preceq \overline{B}_i$, note that for any $t \in \Tset$,
\begin{subequations}
\begin{align}
\underline{B}_{i, t} &= \overline{C}_{i,t} + \max_{\tau \ge t^\prime} \,\{\underline{S}_{i,\tau} - \overline{C}_{i,\tau}\}\notag\\
& \ge \overline{C}_{i,t} + (\underline{S}_{i,t} - \overline{C}_{i,t}) = \underline{S}_{i,t}\\
\overline{B}_{i, t} &= \underline{C}_{i,t} + \min_{\tau \ge t^\prime} \,\{\overline{S}_{i,\tau} - \underline{C}_{i,\tau}\}\notag\\
& \le \underline{C}_{i,t} + (\overline{S}_{i,t} - \underline{C}_{i,t}) = \overline{S}_{i,t}
\end{align}
\end{subequations}
therefore, $p_i \in \X_i$. To complete the proof, we need to show such $p_i$ exists. We claim that if $\underline{B}_i \preceq \overline{B}_i$, Alg.\ref{alg:feas_oracle} always returns a feasible solution $p_i$.
Note that the meaning of $S_{i,t}$ in Alg.\ref{alg:feas_oracle} is consistent with its definition here, and the $\textsf{clip}$ operator in line 8 ensures that
$$
\underline{B}_{i,t} \le \underline{\beta}_{i,t} \le S_{i,t} \le \overline{\beta}_{i,t} \le \overline{B}_{i,t},~~\forall t \in \Tset
$$
We need to further verify $\forall t \in \Tset$:
(i) $\underline{\beta}_{i,t} \le \overline{\beta}_{i,t}$, so that $\textsf{clip}$ in line 8 is valid, and (ii) $\underline{P}_{i,t} \le p_{i,t} \le \overline{P}_{i,t}$, so that $p_i$ is indeed feasible.

For (i), when $t=1$, it is easy to verify with the fact that $\underline{G}_{i,t} \le 0 \le \overline{G}_{i,t}$ using \eqref{eq:feas_necc}. When $t>1$,
\begin{align}
&S_{i,t-1} + \underline{P}_{i,t} \le \overline{\beta}_{i,t-1} + \underline{P}_{i,t}\notag\\
& \quad \le \overline{B}_{i,t-1} + \underline{P}_{i,t}
= \underline{C}_{i,t-1} + \overline{G}_{i,t-1} + \underline{P}_{i,t}\notag\\
& \quad = \underline{C}_{i,t} + \overline{G}_{i,t-1}
\le \underline{C}_{i,t} + \overline{G}_{i,t} = \overline{B}_{i,t}
\end{align}
Similarly, we can show that $S_{i,t-1} + \overline{P}_{i,t} \ge \underline{B}_{i,t}$. Together, this proves $\underline{\beta}_{i,t} \le \overline{\beta}_{i,t}$.

For (ii), it is straightforward by noticing that
\begin{align}
p_{i,t} = S_{i,t} - S_{i,t-1} \ge \underline{\beta}_{i,t} - S_{i,t-1} \ge \underline{P}_{i,t}
\end{align}
and similarly, $p_{i,t} \le \overline{P}_{i,t}$.
\end{enumerate}

The time complexity of Alg.\ref{alg:feas_oracle} is obvious, as for a generic $x \in \R^T$, $H x$ can be computed in $\mathcal{O}(T)$ by a forward loop,
and $\max_{\tau \ge t} x_\tau$ can be computed in $\mathcal{O}(T)$ by a backward loop.
This completes the proof.
\end{proof}

\section{Details on the Best-Response Solver and Its (S1) Variant}
\label{appx:s1_solver}

\subsection{Adaptation to the ADMM Subproblem (S1)}\label{appx:S1_variant}
The per-EV update \hyperref[eq:step_1]{(S1)} of the ADMM baseline reduces to the structure of \eqref{eq:D1_exp} as follows. First note that
\begin{equation}
    \norm{A_i p_i - b_i^{(k)}}_F^2 = \norm{p_i - \widetilde{b}_i^{(k)}}_2^2
\end{equation}
where $\widetilde{b}_i^{(k)} \in \R^T \coloneqq A_i^\top b^{(k)}_i = p_i^{(k)} - \frac{1}{I} A_i^\top (\hat{l}^{(k)}  - l^{(k)} + I{\mu}^{(k)})$.
Therefore, with $\Jev_i(p_i) = \frac{\kappa}{2}\norm{p_i}_2^2$ and normalizing the objective by $\rho$,
\hyperref[eq:step_1]{(S1)} seeks to solve:
\begin{subequations}\label{eq:step_1_exp}
\begin{align}
\min_{p_i}~~& \frac{1}{2}\norm{p_i - \widetilde{b}_i}_2^2 + \frac{\kappa_\rho}{2}\norm{p_i}_2^2\\
\text{s.t.}~~& \underline{P}_i \preceq p_i \preceq \overline{P}_i\\
& \underline{S}_i \preceq Hp_i \preceq \overline{S}_i
\end{align}
where $\kappa_\rho = \kappa / \rho > 0$.
\end{subequations}
The objective is a $(1+\kappa_\rho)$-strongly convex quadratic over the same constraints as \eqref{eq:D1_exp}.
Consequently, Alg.~\ref{alg:step1} applies with the closed form
\begin{equation}\label{eq:p_dagger_S1}
p^\dagger(u_i) = \text{clip}\left(
\frac{1}{1+\kappa_\rho}\left[\widetilde{b}_i + H^\top (\underline{u}_i - \overline{u}_i)\right],
\underline{P}_i, \overline{P}_i
\right)
\end{equation}
in place of \eqref{eq:p_dagger_D1}, and Theorem~\ref{thm:S1_conv} holds with $L_\psi=\frac{2}{1+\kappa_\rho}\norm{H}_2^2$ and $\widetilde{\eta} = \eta(1+\kappa_\rho)$, as shown in the proof below.

\subsection{Closed-form solution of \texorpdfstring{$p^\dagger(u_i)$}{p†(ui)} in \eqref{eq:p_dagger_D1} / \eqref{eq:p_dagger_S1}}
\begin{proof}
Consider the minimization problem of
\begin{subequations}\label{eq:min_quad_box}
\begin{align}
\min_{x\in \R^n}\, & h(x) \coloneqq \beta_1 \norm{x - \alpha_1}_2^2 + \beta_2 \norm{x}_2^2 + \beta_3\langle x, \alpha_3 \rangle\\
\text{s.t.}~~ & \underline{x} \preceq x \preceq \overline{x}
\end{align}
where $\beta_1 + \beta_2 > 0$.
\end{subequations}
Note that both objective and the box constraint are separable element-wisely.
In particular, $h(x) = \sum_j h_j(x_j)$, where
\begin{align}
    h_j(x_j) & = \beta_1 (x_j - \alpha_{1, j})^2 + \beta_2 x_j^2 + \beta_3 x_j  \alpha_{3,j}\\
             & = \underbrace{(\beta_1 + \beta_2)}_A x_j^2 + B x_j + C\notag
\end{align}
with constants (w.r.t. $x_j$) $B$ and $C$.
Clearly, $x^\star_j = -\frac{B}{2A}$ is the unique minimizer of the unconstrained function $h_j(x_j)$.
Therefore, the unique minimizer of $\min\{h_j(x_j): x_j \in [\underline{x}_j, \overline{x}_j]\}$ is $\text{clip}\left(x_j^\star, \underline{x}_j, \overline{x}_j\right)$.

This shows, suppose $x^\star$ is the unique minimizer of $h(x)$ (i.e., $\nabla h(x^\star) = 0$), then $\text{clip}(x^\star, \underline{x}, \overline{x})$ is the unique minimizer of \eqref{eq:min_quad_box}.

To complete the proof, simply note that the Lagrangian $\mathcal{L}(p_i, u_i)$ for both (S1) and (D1) can be written in the form of $h(x)$ (with $x = p_i$).
\end{proof}

\subsection{Proof of Theorem~\ref{thm:S1_conv}}
Theorem.\ref{thm:S1_conv} applies general theorems to the problem \eqref{eq:D1_exp} and its (S1) variant \eqref{eq:step_1_exp}.
The original theorems \cite[Ch.12]{beck2017first} are about proximal gradient descent algorithms, but it is well-known that projected gradient descent is a special case of proximal descent.

The model discussed in \cite[Ch.12]{beck2017first} is:
\begin{equation}\label{eq:Beck_Ch12}
    \min_{x\in \mathbb{E}} \left\{f(x) + g(\mathcal{A}(x))\right\}
\end{equation}
with assumptions:
\begin{itemize}
    \item[(A1)] $f: \mathbb{E} \mapsto (-\infty, +\infty]$ is proper, closed, and $\sigma$-strongly convex;
    \item[(A2)] $g: \mathbb{V} \mapsto (-\infty, +\infty]$ is proper, closed, and convex;
    \item[(A3)] $\mathcal{A}: \mathbb{E} \mapsto \mathbb{V}$ is a linear transformation;
    \item[(A4)] There exists $\hat{x} \in \text{ri}(\text{dom}\, f)$ and $\hat{z} \in \text{ri}(\text{dom}\,g)$ such that $\mathcal{A}(\hat{x}) = \hat{z}$.
\end{itemize}
Suppose the assumptions hold, both the dual objective sequence \cite[Thm 12.4]{beck2017first} and the primal point sequence \cite[Thm 12.8]{beck2017first} generated by the dual proximal gradient method (DPG) converge in $\mathcal{O}(1/k)$
with step size $\frac{1}{L}$ where $L \ge L_F = \frac{\norm{\mathcal{A}}^2}{\sigma}$.

Meanwhile, both \eqref{eq:D1_exp} and \eqref{eq:step_1_exp} minimize a $\sigma$-strongly convex quadratic $f_0(p_i)$ over the box and cumulative-energy constraints of $\X_i$, with $\sigma = \kappa$ for \eqref{eq:D1_exp} and $\sigma = 1+\kappa_\rho$ for \eqref{eq:step_1_exp}.
Both can be written in the form of \eqref{eq:Beck_Ch12} by recognizing
\begin{subequations}
\begin{align}
    &\mathbb{E}  = \{x \in \R^T: \underline{P}_i \preceq x \preceq \overline{P}_i\}, ~~\mathbb{V} = \R^{2T}\\
    &f(x)  = f_0(x), \quad \mathcal{A} : \R^T\mapsto \R^{2T} = \begin{bmatrix} H \\ - H    \end{bmatrix}\\
    &g(z)  = \Ind_{\mathcal{C}}(z), ~\mathcal{C} = \{z\in \R^{2T}: z \preceq \begin{bmatrix} \overline{S}_i \\ - \underline{S}_i    \end{bmatrix}\}
\end{align}
\end{subequations}
It is easy to see that (1) $f$ is proper, closed, and $\sigma$-strongly convex;
(2) $g$ is proper (if $\mathcal{C}\neq \emptyset$), closed, and convex;
(3) $\mathcal{A}$ is a linear transformation; and
(4) suppose $\X_i \neq \emptyset$, any $\hat{x} \in \X_i$ and $\hat{z} = \mathcal{A}(\hat{x})$, $\hat{x} \in \text{ri}(\text{dom}\, f)$ and $\hat{z} \in \text{ri}(\text{dom}\,g)$.

Lastly, for step size choice,
\begin{align*}
    \norm{\mathcal{A}}^2 = 2\norm{H}^2 = 2\lambda_{\max}(H^\top H) = 2\lambda_{\max}(H H^\top)
\end{align*}
Note the matrix $\left(H H ^\top\right)_{ij} = \min\{i,j\}$. According to \cite{da2007eigenvalues},
\begin{equation*}
    \lambda_{\max}(H H^\top) = \frac{1}{4\sin^{2}\frac{\pi}{4T+2}}
\end{equation*}
Further, using Taylor expansion of $\sin^{-2}(x)$ at $x_0=0$,
we can approximate $\lambda_{\max}(H H^\top) \approx \frac{4}{\pi^2} T^2 \approx 0.405 T^2$.

This completes the proof.

\subsection{Implementation of Projected Adam Optimizer}
\label{appx:proj_adam}

Alg.~\ref{alg:step1} uses vanilla projected gradient ascent with constant step size $\eta$, for which convergence is established in Theorem~\ref{thm:S1_conv}.
In practice, we replace the gradient ascent step (line 6) with \textsf{Adam}, an adaptive first-order method that maintains exponential moving averages of the gradient and its squared magnitude.
Because the dual variable must satisfy $u_i \succeq 0$, we adopt the \emph{projected} variant of \textsf{Adam}. At each iteration the gradient is first projected onto the feasible cone before entering the moment updates, and the final iterate is projected again after the ascent step.
Concretely, denote the raw gradient $g^{\langle k \rangle} = \nabla\psi(u_i^{\langle k \rangle})$ and define the \emph{projected gradient}
\begin{equation}\label{eq:proj_grad}
    \widetilde{g}^{\langle k \rangle} = \frac{1}{\varepsilon}\left(\left[u_i^{\langle k \rangle} + \varepsilon\, g^{\langle k \rangle}\right]^+ - u_i^{\langle k \rangle}\right), \quad \varepsilon = 10^{-8}
\end{equation}
which equals $g^{\langle k \rangle}$ whenever $u_i^{\langle k \rangle}$ is in the strict interior of the non-negative orthant, and zeroes out components that would push the iterate below zero.

The moment updates and step then proceed as follows.
With hyper-parameters $\beta_1 = 0.9$, $\beta_2 = 0.999$, and base learning rate $\eta = 1$:
\begin{subequations}\label{eq:adam_updates}
\begin{align}
    m^{\langle k+1 \rangle} &= \beta_1\, m^{\langle k \rangle} + (1 - \beta_1)\, \widetilde{g}^{\langle k \rangle}\label{eq:adam_m}\\
    v^{\langle k+1 \rangle} &= \beta_2\, v^{\langle k \rangle} + (1 - \beta_2)\, (\widetilde{g}^{\langle k \rangle})^2\label{eq:adam_v}\\
    \hat{m}^{\langle k+1 \rangle} &= \frac{m^{\langle k+1 \rangle}}{1 - \beta_1^{k+1}}, \qquad
    \hat{v}^{\langle k+1 \rangle} = \frac{v^{\langle k+1 \rangle}}{1 - \beta_2^{k+1}}\label{eq:adam_bias}\\
    u_i^{\langle k+1 \rangle} &= \left[u_i^{\langle k \rangle} + \frac{\eta}{\sqrt{\hat{v}^{\langle k+1 \rangle}} + \varepsilon} \odot \hat{m}^{\langle k+1 \rangle}\right]^+\label{eq:adam_step}
\end{align}
\end{subequations}
where all operations are element-wise, $m^{\langle 0 \rangle} = v^{\langle 0 \rangle} = 0$, and $(\cdot)^2$ denotes the element-wise square.
After the dual update \eqref{eq:adam_step}, the primal variable is recovered analytically via $p_i^{\langle k+1 \rangle} = p^\dagger(u_i^{\langle k+1 \rangle})$ as in \eqref{eq:p_dagger_S1}/\eqref{eq:p_dagger_D1}.
Convergence is monitored via the relative duality gap $\Delta$ defined in \eqref{eq:S1_rel_duality_gap}, with tolerance $\epsilon = 10^{-3}$.

\smallskip
\noindent\textbf{Active-set masking.}\quad
When solving (S1) for all $I$ EVs in parallel, the per-EV subproblems converge at different rates.
To avoid wasting computation on already-converged EVs, we maintain a boolean mask over the EV index set:
every 20 inner iterations, the relative duality gap is evaluated, and EVs with $\Delta \le \epsilon$ are removed from the active set.
Subsequent iterations operate only on the remaining active EVs, reducing both computation and memory.

\smallskip
\noindent\textbf{Empirical comparison.}\quad
Fig.~\ref{fig:optimizer_comparison_D1} compares the convergence of three optimizer variants on the (D1) subproblem instances of Sec.~\ref{sec:solver_analysis}:
\begin{enumerate}
    \item Vanilla projected gradient ascent (GD) with constant step size $\eta=10^{-7}$, of the order $1/L_\psi$ suggested by Theorem~\ref{thm:S1_conv};
    \item Projected GD with a normalized diminishing step size $\eta_k = \eta / (\norm{g^{\langle k \rangle}}_2 \sqrt{k+1})$ with $\eta=10^{-4}$, which ensures convergence even when only subgradients are available \cite[Sec.~8.2]{beck2017first}; and
    \item Projected Adam \eqref{eq:adam_updates} with $\eta=10^{-3}$.
\end{enumerate}

Among the three, Adam converges substantially faster in terms of both median and tail behavior, motivating its use as the default optimizer in our implementation.

\begin{figure}[h!]
    \centering
    \includegraphics[width=\columnwidth]{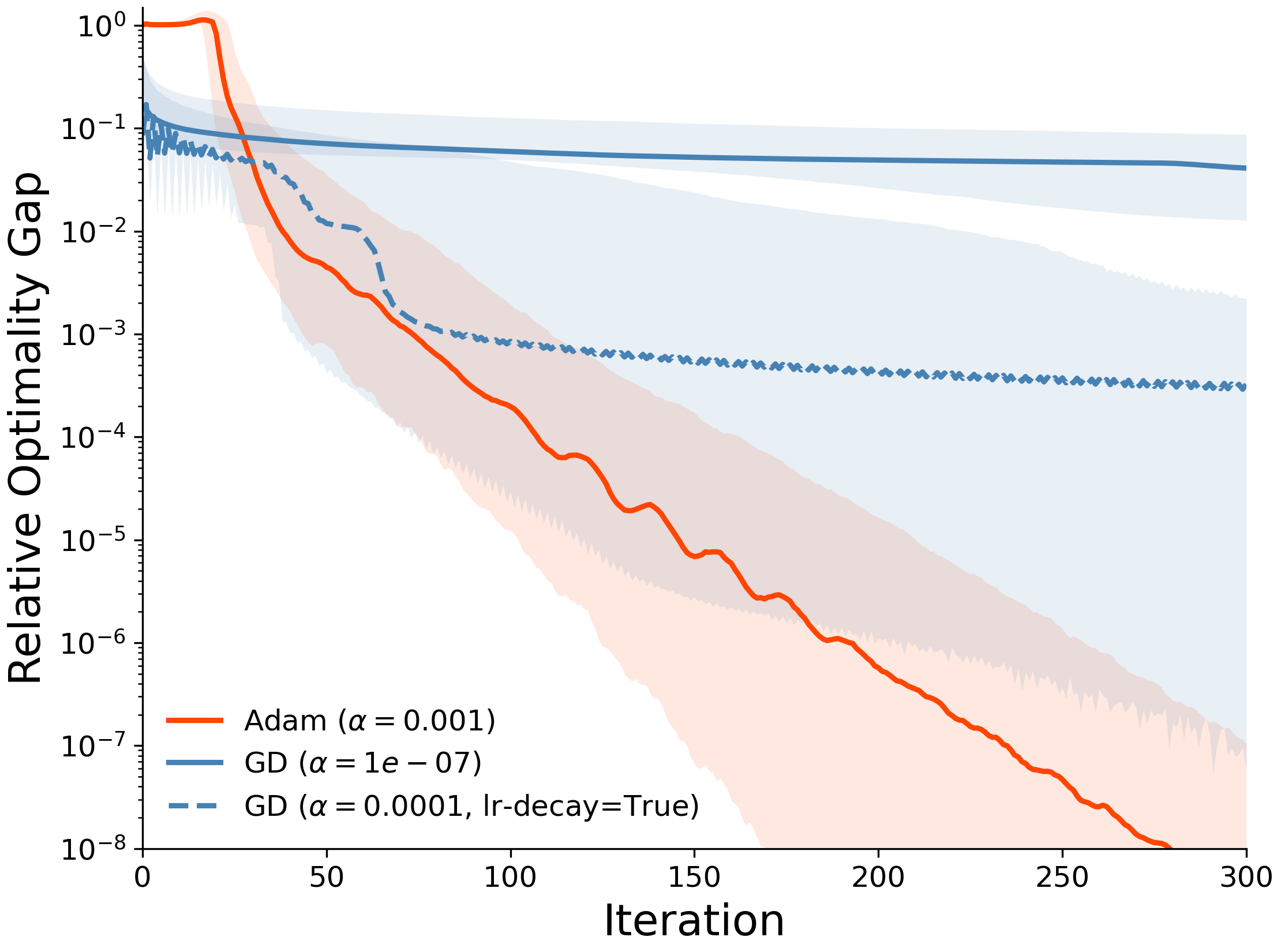}
    \caption{\textbf{Optimizer comparison (D1)}
    \scriptsize\quad
    Medians (solid line) and [10\%, 90\%] percentiles (shaded) of evaluated relative optimality gaps over different optimizers.
    }
    \label{fig:optimizer_comparison_D1}
\end{figure}

\section{Grid-Side Subproblems}
\label{appx:S2_solver}

\subsection{Proof of Lemma~\ref{lem:D2}}

Recall from \eqref{eq:D2} that $\text{(D2)} = \min_l \Jgrid(l) - \langle \lambda, l \rangle$, which decomposes per feeder~$s$ as:
\begin{equation}\label{eq:appx_D2}
    \min_{v_s \ge 0,\, l_s}\; v_s - \langle \lambda_s, l_s \rangle \quad \text{s.t.}\quad l_{s,t} - C_{s,t} \le v_s,\; \forall t
\end{equation}
Note that $l_s$ is unconstrained below (only upper-bounded by $C_s + v_s \boldsymbol{1}$).

\begin{proof}
We show the optimal objective of \eqref{eq:appx_D2} is $-\langle \lambda_s, C_s \rangle$ when $\lambda_s \succeq 0$ and $\sum_t \lambda_{s,t} \le 1$, and $-\infty$ otherwise. Summing over $s \in \Sset$ then proves the lemma.

\smallskip
\noindent\emph{Case 1: $\min_t \lambda_{s,t} < 0$.}\quad
Pick $t_0$ with $\lambda_{s,t_0} < 0$. Fix $v_s = 0$ and $l_{s,t} = 0$ for $t \neq t_0$.
Since $l_{s,t_0}$ is only upper-bounded, sending $l_{s,t_0} \to -\infty$ gives $-\lambda_{s,t_0}\, l_{s,t_0} \to -\infty$, so \eqref{eq:appx_D2} is unbounded below.

\smallskip
\noindent\emph{Case 2: $\lambda_s \succeq 0$ and $\sum_t \lambda_{s,t} > 1$.}\quad
Since $\lambda_s \succeq 0$, minimizing $-\langle \lambda_s, l_s \rangle$ pushes $l_s$ to its upper bound, $l_{s,t} = C_{s,t} + v_s$ for all $t$.
The objective becomes
$
    v_s\!\left(1 - \textstyle\sum\nolimits_t \lambda_{s,t}\right) - \langle \lambda_s, C_s \rangle.
$
Since $1 - \sum_t \lambda_{s,t} < 0$, sending $v_s \to +\infty$ gives objective $\to -\infty$.

\smallskip
\noindent\emph{Case 3: $\lambda_s \succeq 0$ and $\sum_t \lambda_{s,t} \le 1$.}\quad
We claim $v_s^\star = 0$, $l_s^\star = C_s$ is optimal, with objective $-\langle \lambda_s, C_s \rangle$.
This point is feasible, since $v_s = 0 \ge 0$ and $l_{s,t} - C_{s,t} = 0 \le v_s$.
For any feasible $(v_s, l_s)$, since $\lambda_s \succeq 0$ and $l_{s,t} \le C_{s,t} + v_s$:
\begin{align}
    v_s - \langle \lambda_s, l_s \rangle & \ge v_s - \textstyle\sum\nolimits_t \lambda_{s,t}(C_{s,t} + v_s)\\
    & = v_s\!\left(1 - \textstyle\sum\nolimits_t \lambda_{s,t}\right) - \langle \lambda_s, C_s \rangle \ge -\langle \lambda_s, C_s \rangle\notag
\end{align}
where the last step uses $v_s \ge 0$ and $1 - \sum_t \lambda_{s,t} \ge 0$.

\smallskip
Combining all cases proves Lemma~\ref{lem:D2}.
Checking $\lambda_s \succeq 0$, $\sum_t \lambda_{s,t} \le 1$, and evaluating $-\langle \lambda_s, C_s \rangle$ each require a single pass over $T$ entries, i.e., $\mathcal{O}(T)$ per feeder.
\end{proof}

\subsection{The (S2-$s$) Solver of the ADMM Baseline}

With $\Jgrid(l)$ as in \eqref{eq:Jgrid}, the grid update \hyperref[eq:step_2]{(S2)} of the ADMM baseline decomposes per feeder $s$ as follows, with $\varrho \coloneqq \rho / I$:
\begin{subequations}\label{eq:S2s}
\begin{align}
    (\text{S2-}s)\qquad \min_{v_s, l_s}~ & v_s + \frac{\varrho}{2}\norm{l_s - D_s}_2^2\\
    \text{s.t.}~ & -v_s \le 0 \label{eq:appx_S2_constr1}\\
    & l_{s,t} - C_{s,t} \le v_s, ~\forall t \label{eq:appx_S2_constr2}
\end{align}
\end{subequations}
where $v_s \in \R$ and $l_s \in \R^T$.

\begin{theorem}\label{thm:S2}
Subproblem \hyperref[eq:S2s]{(S2-$s$)} can be solved exactly by a sort-and-search procedure (Alg.~\ref{alg:alg_step2}) with time complexity $\mathcal{O}(T \ln T)$.
\end{theorem}

\begin{algorithm}
\caption{Custom solver for \hyperref[eq:S2s]{(S2-$s$)} subproblem}\label{alg:alg_step2}
\begin{algorithmic}[1]\small
\Procedure{Feeder\_Update\_S2~}{$D_s$}
\State define $\delta_t\coloneqq D_{s,t} - C_{s,t}$; $\varrho \coloneqq \rho / I$
\State sort $(\delta_t)$ as $\widetilde{\delta}_1 \ge ... \ge \widetilde{\delta}_{T}$
\State $\widetilde{\delta}_{T+1} \coloneqq 0$
\Comment{Avoid ``out of index'' at line 11}
\State define $S_k \coloneqq \sum_{j=1}^k [\widetilde{\delta}_j]^+$
\If{$\varrho\,S_{T} \le 1$}
\State $v_s \gets 0$, ~$\nu^\circ_s \gets 1 - \varrho\,S_{T}$
\Else
\For{$k = 1, ..., T$}
\State $v^\star \gets \frac{1}{k}\left(S_k - \varrho^{-1} \right)$
\If{$\widetilde{\delta}_{k+1} < v^\star \le \widetilde{\delta}_k$}
\State $v_s\gets v^\star$, $\nu^\circ_s \gets 0$
\State \textbf{break}
\EndIf
\EndFor
\EndIf
\State $l_{s,t} \gets \min\{C_{s,t}+v_s, D_{s,t}\}, ~\nu_{s,t} \gets \varrho\,[\delta_{t}-v_s]^+$
\State \Return  $v_s, l_s, \nu^\circ_s, \nu_s$
\State \Comment{$\nu^\circ_s \in \R, \nu_s \in \R^T$ are the optimal dual variables of \hyperref[eq:S2s]{(S2-$s$)}}
\EndProcedure
\end{algorithmic}
\end{algorithm}

\begin{proof}[Proof sketch of Theorem~\ref{thm:S2}]
Let $\nu_s \in \R^T$ and $\nu_s^\circ \in \R$ be the multipliers of \eqref{eq:appx_S2_constr2} and \eqref{eq:appx_S2_constr1}, respectively, and define $\delta_{s,t} \coloneqq D_{s,t} - C_{s,t}$.
Stationarity in $l_s$ gives $\nu_{s,t} = \varrho\,(D_{s,t} - l_{s,t})$, and complementary slackness then determines, for any fixed $v_s \ge 0$,
$l_{s,t} = \min\{C_{s,t} + v_s,\, D_{s,t}\}$ and $\nu_{s,t} = \varrho\,[\delta_{s,t} - v_s]^+$.
The remaining KKT conditions reduce to the scalar system
$\nu^\circ_s + \Xi(v_s) = 1$, $\nu^\circ_s\, v_s = 0$, $\nu^\circ_s \ge 0$, $v_s \ge 0$,
where $\Xi(v_s) \coloneqq \varrho \sum_t [\delta_{s,t} - v_s]^+$ is continuous, non-increasing, and strictly decreasing wherever positive.
Hence, if $\Xi(0) \le 1$, the unique solution is $v_s = 0$ with $\nu^\circ_s = 1 - \Xi(0)$. Otherwise the solution has $v_s > 0$ solving $\Xi(v_s) = 1$, which is unique by strict monotonicity, and sorting $\{\delta_{s,t}\}_t$ locates it exactly in one scan over the breakpoint intervals, which is what lines 6--13 of Alg.~\ref{alg:alg_step2} implement.
Uniqueness of the KKT point together with convexity of \eqref{eq:S2s} yields correctness \cite[p.244]{boyd2009convex}. The sort dominates the cost, hence $\mathcal{O}(T \ln T)$.
\end{proof}

\smallskip
\noindent\textbf{Empirical comparison.}\quad
Table~\ref{tab:S2_timing} compares the wall-clock time of Alg.~\ref{alg:alg_step2} against CVXPY on a single (S2-$s$) instance ($T=168$).
``serial'' means solving all $S=1300$ feeders in a Python loop. For Alg.~\ref{alg:alg_step2}, ``batch'' means vectorizing the sort-and-search procedure across feeders using NumPy, and for CVXPY it means solving (S2) as a whole optimization problem.
Alg.~\ref{alg:alg_step2} achieves a speedup of roughly three orders of magnitude in the serial setting and over five orders of magnitude in the batched setting (stacking subproblems for CVXPY shows a negative impact on performance). Alg.~\ref{alg:alg_step2} makes the (S2) step negligible relative to (S1) in overall ADMM runtime.

\begin{table}[h!]
    \centering
    \begin{tabular}{c|c c c}
    \toprule
        solver &  $S=1$ & serial ($S=1300$)& batch ($S=1300$)\\
    \midrule
        CVXPY & $114$ ms & $2$ m $30$ s & $2$ h $45$ m \\
        Alg.\ref{alg:alg_step2} & $0.04$ ms & $52$ ms & $9.81$ ms\\
    \bottomrule
    \end{tabular}
    \caption{\textbf{(S2) solver timing comparison}}
    \label{tab:S2_timing}
\end{table}

\section{Population and Trajectory Synthesis}\label{appx:data_synth}

\subsection{Sampling future EV adopters}\label{appx:adopter_sampling}

Present-day EV adopters differ substantially from the general driver population in age, income, education, and housing, so sampling new adopters uniformly at random would misrepresent how the adopter pool evolves as penetration grows. Instead, we let the demographic composition of EV drivers follow a gradual transition. At a target adoption rate $p$, the number of adopters in each demographic subgroup is linearly interpolated between today's EV population and a fully electrified fleet ($p=100\%$, at which adopters match all drivers). Sampling is hierarchical, first allocating new adopters across counties and then across demographic subgroups, defined by the Cartesian product of age, race and ethnicity, income, education, housing type, and homeownership, with each level following the same interpolation. To keep scenarios nested and comparable, adopters are drawn incrementally in $10\%$ steps, so that the $30\%$ adopter set is contained in the $40\%$ set, and so on.

\subsection{Week-long itinerary synthesis}\label{appx:traj_synth}

The mobility data provide disaggregate travel diaries for two representative days (a Thursday and a Saturday). We extend these into logically coherent week-long itineraries for every synthetic individual through the following cluster-based bootstrapping procedure.

\emph{Cohort grouping.} Individuals are first partitioned by immutable attributes (home block group, work block group, and sex), so that all subsequent sampling occurs within demographically and spatially comparable cohorts.

\emph{Similarity.} Within each cohort, we form a pairwise similarity matrix that combines (i) continuous behavioral descriptors (e.g., weekly mileage and longest single trip), min--max normalized and weighted, and (ii) categorical socio-demographics (e.g., race and ethnicity), one-hot encoded and weighted. Cosine similarity of the combined feature vectors yields an $n\times n$ matrix in which larger entries indicate closer behavioral proximity.

\emph{Assembly.} For each individual, daily diaries drawn from close peers (sampled with probability increasing in similarity) are concatenated into a seven-day sequence. A candidate week is accepted only if it exhibits no temporal overlaps or spatial inconsistencies between consecutive stays and trips, and is resampled otherwise. Accepted itineraries are exported at hourly resolution and drive the energy-consumption terms $D_{i,t}$ used in the coordination problem.

\section{Baseline Policies}\label{appx:baselines}

We detail the baselines benchmarked against \textsf{MAC} in Sec.~\ref{sec:baselines}, where the evaluation metrics are also defined. Throughout, the \emph{stays} of EV $i$ are its maximal parked intervals at which charging is possible. Let stay $j$ occupy the slot set $\mathcal{T}_j \subseteq \Tset$ (so $\overline{P}_{i,t}>0$ for $t\in\mathcal{T}_j$), with delivered energy $q_j = \Delta t\sum_{t\in\mathcal{T}_j} p_{i,t} \in [0,\overline{Q}_j]$ and $\overline{Q}_j = \Delta t\sum_{t\in\mathcal{T}_j}\overline{P}_{i,t}$.

\subsection{Intra-stay dispatch}
A policy first sets a per-stay energy $\{q_j\}$, and a \emph{dispatch} rule then maps it to a slot-level profile $p_i$, from which feeder loads follow via \eqref{eq:feeder_agg_st}. We use two dispatches. \emph{Front-loaded}: fill the slots of $\mathcal{T}_j$ at $\overline{P}_{i,t}$ from the earliest until $q_j$ is exhausted. \emph{Uniform}: spread $q_j$ in proportion to availability, $p_{i,t}=q_j\,\overline{P}_{i,t}/\overline{Q}_j$ for $t\in\mathcal{T}_j$. Front-loading is the default, as it upper-bounds the hourly peak. The uniform variant is used as a sensitivity in the sub-hourly discussion.

\subsection{Rule-based policies}
\noindent\textbf{\textsf{ASAP+}.}\quad Processing stays in temporal order, let $e^{\mathrm{arr}}_j$ be the energy in the battery on arrival at stay $j$, and $R_j$ the energy still required for the remaining trips to be feasible and for the terminal balance $e_{i,T}=E_{i,0}$ to hold. If the EV arrives below half capacity, $e^{\mathrm{arr}}_j < 0.5\,E^{\max}_i$, it charges as much as it still needs, $q_j = \min\{\overline{Q}_j,\, E^{\max}_i - e^{\mathrm{arr}}_j,\, R_j\}$. Otherwise it charges only the minimum that keeps the remaining itinerary feasible, which is zero at most stays and positive only ahead of a long trip or of the terminal balance. In both cases the terminal target caps the total, so the weekly energy equals the minimum feasible $\sum_t D_{i,t}$. The policy delivers no surplus, and the arrival threshold decides only when that energy is drawn. The intermediate per-stay state-of-charge floors are low and never raise the total.

\smallskip
\noindent\textbf{\textsf{ASAN}.}\quad Let $\mathrm{nx}(j)$ be the first stay after $j$ with duration at least $\underline{\tau}=3$~h, and $\Delta_j$ the driving energy consumed between stay $j$ and $\mathrm{nx}(j)$. \textsf{ASAN} charges the least of a full battery, the energy to reach $\mathrm{nx}(j)$ holding a $50\%$ reserve, and the energy the terminal target still requires:
\begin{equation*}
q_j = \min\Big\{\mathrm{clip}\!\left(\Delta_j + 0.5\,E^{\max}_i - e^{\mathrm{arr}}_j,\; 0,\; \overline{Q}_j\right),\; R_j\Big\},
\end{equation*}
where $e^{\mathrm{arr}}_j$ and $R_j$ are the arrival energy and the remaining requirement defined for \textsf{ASAP+} above. The reserve rule uses only local look-ahead. The $R_j$ cap keeps \textsf{ASAN} energy-balanced, so it delivers the same weekly energy as \textsf{ASAP+} and differs only in which stays carry it.

\subsection{Optimization-based policies}
\noindent\textbf{\textsf{minPeak}.}\quad A mobility-aware but spatially blind policy in which each EV \emph{individually} minimizes the peak of its own charging profile,
\begin{equation*}
p_i^{\textsf{minPeak}} = \arg\min_{p_i\in\X_i}\ \max_{t\in\Tset}\ p_{i,t}, \qquad \forall i \in \Iset,
\end{equation*}
solved per EV by bisection on the peak level using the $\mathcal{O}(T)$ feasibility oracle (Thm.~\ref{thm:feas_oracle}). Each EV spreads its own demand as flat as its mobility allows, so the summed profile acts as a greedy heuristic for shaving the aggregate temporal peak. The policy involves no coordination among EVs and ignores the per-feeder capacities $C_{s,t}$.

\smallskip
\noindent\textbf{\textsf{PR}.}\quad A one-shot price response. From the \textsf{ASAP+} load $l^{\mathrm{ref}}=\sum_i A_i p_i^{\textsf{ASAP+}}$, define the shortage $\delta_{s,t}=[l^{\mathrm{ref}}_{s,t}-C_{s,t}]^+$ and a locational price $\lambda_{s,t}\ge 0$ that is dual-feasible ($\sum_t \lambda_{s,t}\le 1$ for each feeder) and concentrated on congested feeder-slots. Each EV best-responds once,
\begin{equation*}
p_i^{\textsf{PR}} = \arg\min_{p_i\in\X_i}\ \Jev_i(p_i) + \langle A_i^\top \lambda,\, p_i\rangle,
\end{equation*}
i.e., a single competitive-equilibrium step \eqref{eq:EV_i_CE} solved by the same custom (D1) solver. A feeder-invariant price is the special case $\lambda_{s,t}=\lambda_t$.

\smallskip
\noindent\textbf{\textsf{FLEX}.}\quad A session-based counterpart of \textsf{MAC}. Fixing the per-stay energy to a rule's allocation $q_j=q_j^{\mathrm{rule}}$ ($\mathrm{rule}\in\{\textsf{ASAP+},\textsf{ASAN}\}$) restricts the feasible set to
\begin{equation*}
\X_i^{\textsf{FLEX}} = \X_i \cap \Big\{ p_i:\ \textstyle\Delta t\sum_{t\in\mathcal{T}_j} p_{i,t} = q_j^{\mathrm{rule}},\ \forall j \Big\},
\end{equation*}
and \textsf{FLEX} solves the same coordinated problem \eqref{eq:opt_SO} over $\{\X_i^{\textsf{FLEX}}\}$ by ADMM. It thus optimizes only the intra-stay timing of a fixed session energy, whereas \textsf{MAC} additionally reallocates energy across the itinerary. Each \textsf{FLEX} row of Table~\ref{tab:baselines} pins the energy of the session rule it is paired with.

\section{Sub-Hourly Disaggregation of the Hourly Solutions}\label{appx:subhourly}

The optimization operates on hourly slots, while stays begin and end at arbitrary minutes, so the violation metrics in Sec.~\ref{sec:baselines} are computed from hourly-averaged feeder loads and may under-report short-duration violations. To quantify this effect, we re-place every policy's solution onto grids of $5$, $10$, and $15$ minutes without re-solving. The raw itineraries are minute-resolved, so each stay's true arrival and departure times are known. Within each hour, the energy a stay charges in that hour is spread uniformly over the portion of the hour during which the EV is actually present, feeder loads are re-aggregated on the finer grid, and $\mathrm{TV}_{\max}$ is recomputed. This uniform-spread placement is a realistic estimate of the instantaneous load. A worst-case placement that delivers each slot's energy at rated power under adversarial sub-slot alignment would instead yield a ceiling that grows as $1/\Delta t$. It assumes every charging EV on a feeder peaks in the same sub-slot, which does not occur in practice, so we report the uniform-spread estimate. As a consistency check, running the procedure on a $60$-minute grid reproduces each policy's reported hourly $\mathrm{TV}_{\max}$ exactly.

Table~\ref{tab:subhourly} reports, for every policy in Table~\ref{tab:baselines}, the resulting increase of $\mathrm{TV}_{\max}$ over its hourly value. Two patterns stand out. First, the front-loaded rules are saturated. \textsf{ASAP+} and \textsf{ASAN} charge at rated power on arrival, leaving no within-hour headroom for the instantaneous load to exceed what the hourly average already reports, so their totals move by less than $0.1\%$. Second, the relative inflation grows with the quality of the hourly solution and is largest for \textsf{MAC}. Its coordinated solution packs feeder loads flat against the hosting-capacity limit, so any within-hour unevenness induced by partial-hour stays spills over the limit. Hourly aggregation therefore slightly flatters coordination rather than uncontrolled charging. The conclusions of Sec.~\ref{sec:case-study} are unaffected. The absolute increase for \textsf{MAC} is about $37$~MW at the finest grid, its violation at $5$~min remains a $99.6\%$ reduction relative to \textsf{ASAP+} at the same resolution, and the ordering of the policies in Table~\ref{tab:baselines} is unchanged at every resolution. The inflation decays monotonically as the grid coarsens from $5$ to $15$ minutes, a useful consistency pattern.

\begin{table}[htbp]
\centering
\caption{\textbf{Sub-hourly disaggregation of $\mathrm{TV}_{\max}$}}
\label{tab:subhourly}
\begin{tabular}{@{}lrrrr@{}}
\toprule
 & hourly & \multicolumn{3}{c}{increase (MW)}\\
\cmidrule(l){3-5}
 & (MW) & $15$~min & $10$~min & $5$~min\\
\midrule
\textsf{ASAP+}  & 12,439.6 & $+4.4$ & $+5.1$ & $+5.9$\\
\textsf{uASAP+} & 3,590.4 & $+8.3$ & $+9.6$ & $+11.4$\\
\textsf{ASAN}   & 12,441.3 & $+5.4$ & $+6.3$ & $+7.2$\\
\textsf{uASAN}  & 1,910.2 & $+8.2$ & $+9.5$ & $+11.2$\\
\textsf{FLEX} (\textsf{ASAP+}) & 1,792.7 & $+217.0$ & $+233.3$ & $+247.9$\\
\textsf{FLEX} (\textsf{ASAN})  & 744.3 & $+229.2$ & $+258.4$ & $+287.4$\\
\midrule
\textsf{minPeak} & 79.4 & $+11.0$ & $+13.3$ & $+21.8$\\
\textsf{PR}      & 102.6 & $+9.0$ & $+11.7$ & $+19.1$\\
\textsf{MAC}     & 8.1 & $+18.6$ & $+23.7$ & $+36.9$\\
\bottomrule
\end{tabular}
\end{table}

\printbibliography[title=Supplementary References]

\end{document}